\documentclass[a4paper,twocolumn,USenglish]{article}

\ifx\directlua\undefined\ifx\XeTeXcharclass\undefined
 \else\RequirePackage[no-math]{fontspec}[2017/03/31]\fi 
 \else\RequirePackage[no-math]{fontspec}[2017/03/31]\fi

\usepackage{amsmath,amssymb,amsfonts}
\usepackage{amsthm}
\usepackage{hyperref}
\hypersetup{
 colorlinks=false,
 pdfborder={0 0 0},
 pdfborderstyle={/S/U/W 0},
}

 \usepackage[vmargin=2cm, hmargin=1.5cm]{geometry} 
 
 \usepackage{lmodern}

\usepackage{titling} 
 \pretitle{\begin{center} \sffamily\huge} 
 
 \posttitle{\par\end{center}\vspace{0em}}
 
 \preauthor{\begin{center}\sffamily\large} 
 \postauthor{\par\end{center}\vspace{0.5em}}
 \predate{} \date{} \postdate{}

 \usepackage{titlesec} 
 \titleformat{\section}{ \sffamily\Large\bfseries}{\thesection}{1em}{} 
 \titleformat{\subsection}{\sffamily\normalsize\bfseries}{\thesubsection}{1em}{}
 \titleformat{\subsubsection}{\sffamily\normalsize\itshape}{\thesubsubsection}{1em}{} 
\usepackage{fancyhdr}
\usepackage{caption}

 

\usepackage[T1]{fontenc}

\usepackage{enumitem}
\usepackage[numbers,square]{natbib}
\usepackage{color}

\usepackage{subfig}
\usepackage{tikz}
\usetikzlibrary{arrows.meta,bending,shadows,shapes.callouts}
\usetikzlibrary{calc} 
\usetikzlibrary{fit}
\usepackage{pgfplots}

\newtheoremstyle{dgdef_var}
 {.5\baselineskip}
 {.5\baselineskip}
 {\normalfont}
 {}
 {\sffamily\itshape}
 {.}
 {.5em}
 {\thmname{$\blacktriangleright$~#1}\thmnumber{~#2}\thmnote{~(#3)}} 

\theoremstyle{dgdef_var}

\newtheorem{definition}{Definition}[section]
\newtheorem{theorem}[definition]{Theorem} 
\newtheorem{remark}[definition]{Remark} 
\newtheorem{example}[definition]{Example} 
 
\newtheorem{lemma}[definition]{Lemma} 
\newtheorem{proposition}[definition]{Proposition} 
\newtheorem{assumption}[definition]{Assumption}

\newcommand{\level}[0]{\upsilon}

\newcommand{\fine}[0]{ \vspace{0pt}\hfill $\triangleleft$}

\newcommand{\absatz}[1]{
 \par\pagebreak[3]\bigskip 
 \noindent \textbf{\textsf{#1}} 
 \\*[0.25em] 
 \noindent\ignorespaces 
}

\begin{document}

\allowdisplaybreaks

\author{Tessina H. Scholl\thanks{\sffamily The author is with the Institute for Automation and Applied Informatics (IAI), Karlsruher Institute of Technology (KIT), Hermann-von-Helmholtz-Platz 1, 76344 Eggenstein-Leopoldshafen, Germany, Email: tessina.scholl@kit.edu}}
\title{A Lyapunov-Based Perspective on Absolute Stability}

\maketitle
\thispagestyle{fancy} 
\absatz{Abstract}
This article presents a unifying perspective on absolute stability concepts. In particular, it develops a Lyapunov-like explanatory framework for a nonscalar circle criterion with its small-gain and strict-passivity special cases. To this end, a general defining inequality for a Lyapunov-like function is proposed that avoids strict definiteness conditions, enabled by a strengthening of the sector constraint. We discuss different ways to derive a quadratic solution: via a linear matrix inequality (LMI), an algebraic Riccati equation, and a matrix equation. By exploiting the Kalman--Yakubovich--Popov (KYP) lemma, classical frequency-domain results are recovered. A passivity-index-based result is derived that simplifies the evaluation. Overall, the presented interrelations may be useful for both analysis and teaching. 
\\[2em]\noindent\textbf{\textsf{Keywords:}} {\itshape absolute stability, Lyapunov, robustness, nonlinear systems, sector constraints, KYP lemma}

\section{Introduction}

The main objective of this article is to present a unifying perspective on absolute stability results, emphasizing the interrelations between various concepts. To this end, a simple and transparent Lyapunov-like explanatory framework for a nonscalar circle criterion with its small-gain and strict-passivity special cases is developed. Explanatory frameworks for absolute stability criteria are nowadays predominantly Plancherel--Parseval-based approaches, which typically arise from an input-output-operator setting \cite{Megretski.1997,Desoer.2009, Willems.1971}. 
However, Lyapunov-based approaches are also possible \cite{Khalil.2002,Willems.1972b,Scherer.2022,Yakubovich.2004,Gusev.2006,Aizerman.1965,Yakubovich.1973,Yakubovich.1971}. Historically, both even emerged in parallel; see, e.g., \cite{Gusev.2006,Liberzon.2006} for extensive literature reviews. Plancherel--Parseval-based approaches including the concept of integral quadratic constraints (IQCs) \cite{Megretski.1997} are more far-reaching, but Lyapunov-based approaches offer a key advantage: regional stability results are still obtained if the derived sector constraints are only met locally~\cite{Khalil.2002}. 

We refer to the presented result as a Lyapunov-like function since the approach in this article does not insist on strict definiteness properties, which are, e.g., encountered in \cite[Thm.~7.1]{Khalil.2002}, \cite[Thm.~1]{Scherer.2022}, \cite[Thm.~2.2]{Yakubovich.2004}, \cite[eq.~11]{Gusev.2006}, and \cite[p.~29]{Aizerman.1965}. As a consequence, the defining inequality for the Lyapunov-like function is particularly simple, and unbounded regional attractivity statements can be obtained. The latter is demonstrated by an example at the end of this article. 
For the special case of linear norm bounds on the perturbation, a theorem that also relies on a {possibly} semidefinite Lyapunov-like function is given in \cite[Thm.~5.6.25]{Hinrichsen.2005}. In contrast, the approach in the present article is more general: it is not limited to a specific sector, nor to a quadratic solution, nor to a particular way of obtaining it, and (since LaSalle is not used) it remains applicable with time-varying perturbation terms. 
 
The key aspect of the proposed construction principle is the introduction of a possibly nonquadratic offset function in the sector description. This offset function induces a refinement of the sector constraint tailored to the desired properties of the derivative of the Lyapunov-like function. The article focuses on attractivity of an equilibrium, to which end \mbox{{class-$\mathcal K$}} functions are used as offset functions. The approach builds on the author's previous work \cite{Scholl.2025}, which, however, considered time-delay systems and relied on the infinite-dimensional KYP lemma. 
{Its relation to finite-dimensional theory is already presented in} \cite[Ch.~6]{Scholl.2024c}, {which served as a foundation for this article}.

The article is structured as follows. 
Sec.~\ref{sec:problem} introduces the considered problem, { and } 
Sec.~\ref{sec:sectorDescription} 
the sector constraints fundamental to the concept of absolute stability. Sec.~\ref{sec:idea} 
describes the construction principle for the Lyapunov-like function. 
As a main result, Sec.~\ref{sec:attractivityConclusion} shows its usage.
 Secs.~\mbox{\ref{sec:ApproachesV}-\ref{sec:FrequencyDomainResults}} explain how to compute such a function and under which conditions it exists. Finally, Sec.~\ref{sec:Example} presents an example before Sec.~\ref{sec:Conclusion} concludes the paper.

\textit{Notation.}
Given $x\in \mathbb R^n$, $\|x\|$ denotes an arbitrary norm on~$\mathbb R^n$ and $\|x\|_2$ the Euclidean norm. Moreover, $0_n\in \mathbb R^n$ and $0_{n\times m}\in \mathbb R^{n\times m}$ (in short $0$) are zero vector and matrix, $I_n\in \mathbb R^{n\times n}$ (in short $I$) is the identity matrix, and $1_{n\times m}$ a matrix of ones. 
The symmetric part of $A$ is $\mathrm{sym}(A)=\tfrac 1 2 (A+A^\top)$, {whereas $\mathrm{skw}(A)=\tfrac 1 2 (A-A^\top)$}. If $A\in \mathbb C^{n\times n}$, then $\mathrm{He}(A)=\tfrac 1 2 (A+A^H)$, where $A^H=\overline {A^\top}$ describes the conjugate transpose. Moreover, $\mu_2(A)=\lambda_{\max}(\tfrac 1 2 (A^H+A))$ is the logarithmic norm w.r.t.\ the spectral norm, which is $\|A\|_2=\sqrt{\lambda_{\max} (A^H A) }$. 
If $Q=Q^H$, $\lambda_{\min(\max)}(Q)$ is the smallest (largest) eigenvalue. 
Positive \mbox{(semi-)}\allowbreak definiteness is 
$Q\succ (\succeq) 0_{n\times n}$, 
implicitly requiring $Q=Q^H$. 
 The interior of a set $\Omega$ is $\mathrm{int}(\Omega)$. 
{The upper right Dini derivative of $V\colon \mathbb R^n\to\mathbb R$ at $x_0\in \mathbb R^n$ along a {classical} solution $x(t;t_0,x_0)$ of $\dot x(t)=f(t,x(t))$ with $x(t_0;t_0,x_0)=x_0$ is {denoted by} $D_{f(t_0,\cdot)}^+ V(x_0)=\limsup_{h \to 0^+} \frac{V(x(t_0+h;t_0,x_0))-V(x_0)}{h}$}. 
The set of {class-$\mathcal K$} functions is {defined by} $\mathcal K=\{\kappa\in C([0,\infty),\mathbb R_{\geq 0}): {\kappa(0)=0},\text{ strictly increasing}\}$.

\section{Problem description} \label{sec:problem}
We consider a nonlinear, possibly uncertain, system  $\dot x(t)=f(t,x(t))$. Well-posedness is addressed by  Assumption~\ref{asm:aLocLip} below. Assuming that   $f(t,0_n)=0_n$ for all   $t\in \mathbb R$, we are interested in properties of the zero equilibrium. The main result of this article concerns attractivity of this equilibrium and, in particular, the regional or global character of this property. In addition, a result on uniform asymptotic stability will also be derived.
\begin{definition}[Attractivity]
Let $x(t)=x(t;t_0,x_0)$ denote a {maximal} solution with initial value $x(t_0)=x_0\in \mathbb R^n$ of a system with an equilibrium at $x=0_n$. The zero equilibrium is \textit{locally attractive} for a given initial time $t_0\in \mathbb R$ if there exists a set $S\subseteq \mathbb R^n$ with $0_n\in \mathrm{int}(S)$ such that 
\begin{align}
\forall x_0\in S: \quad \|x(t;t_0,x_0)\|\to 0 \quad \text{ as } t\to \infty, 
\end{align} 
and 
\textit{globally attractive} if $S=\mathbb R^n$. 
\fine\end{definition}

As a first step, we decompose the system into
\begin{align}\label{eq:ODEg}
\dot{x}(t)=A x(t)+g(t,x(t)), 
\end{align} 
with a nominal linear part $\dot x=Ax$, $A\in \mathbb R^{n\times n}$, and a possibly nonlinear term {$g\colon [t_0,\infty)\times \mathbb R^n\to \mathbb R^n$} 
that is henceforth referred to as \textit{perturbation}. For instance, 
\begin{samepage}\begin{itemize}
\item $g(t,x)$ might be the remainder of a {time-invariant} linearization (otherwise differentiability is not assumed)
\item $g(t,x)$ might involve a saturation term or {related nonlinearities}, or
\item $g(t,x)=\Delta(t) x$, $\Delta(t)\in \mathbb R^{n\times n}$ might address additive uncertainties in the system matrix $A$. 
\end{itemize}\end{samepage}

\absatz{Incorporation of the perturbation structure} 
To incorporate knowledge on the structure of $g$, {two matrices}
\begin{align*}
B\in \mathbb R^{n\times m} \quad \text{and} \quad C\in \mathbb R^{p\times n} 
\end{align*} 
are to be chosen, with $m\leq n$ and $p\leq n$. The matrix $B$ is intended to encode which components of $\dot x$ are affected by the perturbation. The matrix $C$ is intended to encode on which components (or linear combinations of components) the perturbation relies. Based on the choice of these matrices (which is nonunique and forms a degree of freedom), we can focus on the corresponding nonlinear or uncertain core perturbation function $a\colon [t_0,\infty)\times \mathbb R^p\to \mathbb R^m$ from the decomposition\footnote{The negative sign in (\ref{eq:ga}) {is intended} to resemble a negative feedback $u=-a(t,Cx)$ as input to $\dot x=Ax+Bu$, which is a classical representation of Lur'e systems \cite{Khalil.2002}.}
\begin{align}\label{eq:ga}
g(t,x)=- Ba(t,Cx). 
\end{align} 

In the end, the
 objective is to obtain constraints on this specifically acting {function $a$} 
under which---with some mild assumptions---the perturbed system {(\ref{eq:ODEg}),} 
\begin{align} \label{eq:ABaC}
\dot{x}(t)=A x(t)-Ba(t,Cx(t)){,}
\end{align}  
can be proven to have an attractive zero equilibrium. 

If no structural information {is to} be incorporated, {choosing $B=C=I_n$ directly addresses} $a(t,x)=-g(t,x)$ from (\ref{eq:ODEg}). 
{The following example shows the non-uniqueness of the decomposition (\ref{eq:ga}). Moreover, it demonstrates {that, even in a structured case, one may choose $p=n$ and construct $C$ as a full-rank square matrix} 
(which will not be presumed in the following, but becomes relevant in Thm.~\ref{thm:ConvergenceFullRankC}).} 
\begin{example}\label{exmp:BaC}
We write $g(t,x)=\tilde g(x)$ ($a(t,\zeta)=\tilde a(\zeta)$) in the time-invariant case. For the example (with saturation)
\begin{align*}
\dot x=A x +\tilde g(x)\quad \text{with } \tilde g(x)=\begin{bmatrix} 0 & \cdots & 0 & -\mathrm{sat}(x_1) \end{bmatrix}^\top, 
\end{align*}
possible decompositions (\ref{eq:ga}) include, for instance, 
\begin{align*}
\tilde g(x)=\left[ \begin{smallmatrix} 0 \\\vdots \\ 0 \\ -\mathrm{sat}(x_1)\end{smallmatrix}\right]&= - \underbrace{\left[\begin{smallmatrix} 0 \\\vdots \\ 0 \\ 1\end{smallmatrix}\right]}_{B} \tilde a(\zeta), \quad 
\begin{aligned}[t]
\tilde a(\zeta)&=\mathrm{sat}(\zeta) \\
\text{with } \zeta &= \underbrace{\left[\begin{smallmatrix} 1 & 0 & \cdots & 0 \end{smallmatrix}\right]}_{C} x
\end{aligned}
\end{align*}
or, if a {rank-$n$} matrix $C$ is desired, 
\begin{align*}
\tilde g(x)
&= - \underbrace{\left[\begin{smallmatrix} 0 \\\vdots \\ 0 \\ 1\end{smallmatrix}\right]}_{B} \tilde a(\zeta), \quad 
\begin{aligned}[t]
\tilde a([\zeta_1,\ldots,\zeta_n]^\top)&=\mathrm{sat}(\zeta_1) 
\\
\text{with } \zeta &= \underbrace{\left[\begin{smallmatrix} 1 & & & \\
& \varepsilon & & \\
&&\ddots &\\
&&&\varepsilon \end{smallmatrix}\right]}_{C} x
\\[-3em]\nonumber
\end{aligned}
\end{align*}
with 
{$\varepsilon\neq 0$ arbitrarily small in magnitude}. 
\fine\end{example}
To ensure well-posedness of the system and to simplify the considerations, we {assume the following.} 
\begin{assumption}\label{asm:aLocLip}
The 
function $a\colon [t_0,\infty)\times \mathbb R^p\to \mathbb R^m$ is continuous in $(t,\zeta)$ and locally Lipschitz in $\zeta$ (uniformly w.r.t.\ $t$). Moreover, $a(t,0_p)=0_m$ for all $t\in [t_0,\infty)$, and $a(t,\zeta)$ is uniformly bounded in $t$ on bounded sets of $\zeta$. 
\fine\end{assumption}

\section{Sector description} \label{sec:sectorDescription}

In the field of absolute stability, constraints on the core perturbation {function $a$ in (\ref{eq:ABaC})} 
take the form of \textit{sector} conditions. 

To simplify the presentation, we focus in this section on time-invariant perturbations
\begin{align}\label{eq:tildea}
a(t,\zeta)=\tilde a(\zeta). 
\end{align}
For time-varying perturbations $a(t,\zeta)$, the constraints are required to hold uniformly for all times $t\geq t_0$.

{The sector describes a set (in fact, a cone) of $(\zeta,\alpha)$ pairs in $\mathbb R^p\times \mathbb R^m$ defined via a quadratic inequality}. 
{The non-gray regions in} {Fig.~\ref{fig:sectorLinearNormBound} and \ref{fig:reducedsector} {are examples of sectors} (in {these and all subsequent figures, we consider $p=m=1$}). 
In the end, a slight modification of this set {(unshaded region in Fig.~\ref{fig:reducedsector})} will characterize} admissible pairs of $\zeta$ 
and corresponding function value $\alpha=\tilde a(\zeta)$.
A key idea in the subsequent analysis is to understand 
{a nondegenerate sector} 
as the \textit{zero superlevel set} of {an} \textit{indefinite quadratic form}, { see Fig.~\ref{fig:plotw}. That is, the sector is the set of those $(\zeta,\alpha)$ for which }
\begin{align} \label{eq:sectorConditionw}
w(\zeta,\alpha)\geq 0, 
\end{align}
 and $w$ is a quadratic form 
\begin{align}\label{eq:wzetaalphaPi}
w(\zeta,\alpha):=\begin{bmatrix} \zeta \\ \alpha\end{bmatrix} ^\top \underbrace{\begin{bmatrix} \Pi_{\zeta\zeta} & \Pi_{\zeta\alpha} \\ \Pi_{\zeta\alpha}^\top & \Pi_{\alpha\alpha} \end{bmatrix}}_{\Pi=\Pi^\top} \begin{bmatrix} \zeta \\ \alpha\end{bmatrix}
\end{align}
{of} {a chosen indefinite 
{(or, as a degenerate case, semidefinite)} 
symmetric matrix} $\Pi\in \mathbb R^{(p+m)\times (p+m)}$. 
The graph of the function $\zeta\mapsto \tilde a(\zeta)$ lies in the sector if and only if
\begin{align}
w(\zeta,\tilde a(\zeta))\geq 0
\end{align} for all $\zeta\in \mathbb R^p$, see Fig.~\ref{fig:zerosuperlevelset}.

\begin{figure}
\centering

\subfloat[Indefinite quadratic form {from (\ref{eq:wLinearNormBound})}
 \label{fig:plotw}
]{
\includegraphics[]{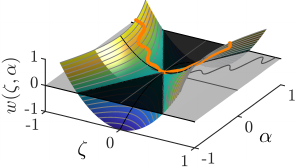}
}\hfill 
\subfloat[Sector {$[-\gamma,\gamma]$} \label{fig:sectorLinearNormBound}
]{ 
\includegraphics[]{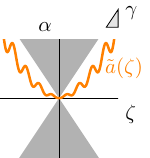}
}
\caption{ {Sector as a superlevel set $\{(\zeta,\alpha)\in \mathbb R^2: 
w(\zeta,\alpha)\geq 0\}$.} 
}
\label{fig:zerosuperlevelset}
\end{figure}

\begin{remark}[Relation to IQCs]
If $t\mapsto (\zeta_{L_2}(t),\alpha_{L_2}(t))$ is square integrable and meets the sector condition (\ref{eq:sectorConditionw}) for all $t\geq t_0$, then $\int_{t_0}^{\infty} w(\zeta_{L_2}(t),\alpha_{L_2}(t))\,\mathrm d t\geq 0$. The latter is an integral quadratic constraint (IQC) \cite{Megretski.1997} with $\Pi$ from (\ref{eq:wzetaalphaPi}) being a constant multiplier. Thus, as an alternative approach to the one in this article, the IQC theorem \cite{Megretski.1997} could be used---at least for global ($\zeta\in \mathbb R^p$) results. 
\fine\end{remark}

\absatz{Sector: Linear norm bound}
The simplest sector constraint {on the perturbation (\ref{eq:tildea}) in (\ref{eq:ABaC})} is a linear bound on the norm of $\tilde a(\zeta)$. That is, 
\begin{align}\label{eq:LinearNormBound}
\|\tilde a(\zeta)\|_2\leq \gamma \|\zeta\|_2 
\end{align}
with a parameter $\gamma\geq 0$. 
In short notation: the graph of $\tilde a$ is in the sector $[-\gamma,\gamma]$ (i.e., the sector bounds have a slope of $\pm\gamma$, see Fig.~\ref{fig:sectorLinearNormBound} for a scalar example). 
This linear norm bound (\ref{eq:LinearNormBound}) can be rewritten as 
\begin{align}
 \gamma^2 \zeta^\top \zeta - \tilde a^\top(\zeta)\tilde a(\zeta) &\geq 0 
\end{align}
or, equivalently (see Fig.~\ref{fig:plotw}), as
\begin{align}
 w(\zeta,\tilde a(\zeta))&\geq 0 \quad \text{with } \nonumber 
\\
 w(\zeta,\alpha)&= \begin{bmatrix} \zeta \\ \alpha \end{bmatrix}^\top
 \underbrace{\begin{bmatrix} \gamma^2 I_p & 0 \\
0 & -I_m
\end{bmatrix} }_{\Pi=\Pi^\top}
 \begin{bmatrix} \zeta \\ \alpha \end{bmatrix}. \label{eq:wLinearNormBound} 
\end{align}
\absatz{Sector: Strict output passivity}
Another important type of constraint is given by the strict output passivity of $\zeta\mapsto \tilde a(\zeta)$ 
with $p=m$, 
\begin{align}\label{eq:passivityineq}
\tilde a^\top\!(\zeta)\, \zeta \geq \rho \|\tilde a(\zeta)\|_2^2 
\end{align}
with output passivity index $\rho>0$. In short notation: the graph of $\tilde a$ is in the sector $[0,1/\rho]$ (i.e., the lower sector bound has a slope of zero and the upper sector bound has a slope of $1/\rho$, see Fig.~\ref{fig:sector_rho}). 
Since 
(\ref{eq:passivityineq}) is equivalent to 
\begin{align}
\tilde a^\top(\zeta) \big(\zeta-\rho \tilde a(\zeta)\big) \geq 0, 
\end{align}
this sector is described by the indefinite quadratic form 
\begin{align}
 w(\zeta,\tilde a(\zeta))&\geq 0 \quad \text{with } \nonumber 
\\*
 w(\zeta,\alpha)&= \begin{bmatrix} \zeta \\ \alpha \end{bmatrix}^\top
 \underbrace{\mathrm{sym}\left( \begin{bmatrix} 0 & 0 \\
I & -\rho I
\end{bmatrix} \right)}_{\Pi=\Pi^\top}
 \begin{bmatrix} \zeta \\ \alpha \end{bmatrix}. 
\label{eq:wOutputPassivity} 
\end{align}
The involved symmetrization relies on the fact that, for any possibly nonsymmetric real matrix~$M$, it holds
\begin{align}\label{eq:sym}
z^\top M z = z^\top \mathrm{sym}(M)z, \quad \mathrm{sym}(M)\smash{\stackrel{\mathrm{def}}= }\tfrac 1 2 (M+M^\top)
\end{align}
{since} $M=\mathrm{sym}(M)+\mathrm{skw}(M)$ and $z^\top \mathrm{skw}(M)z=0$.
\absatz{Sector: General sector constraints}
General sector bounds $[K_1,K_2]$ with lower slope $K_1\in \mathbb R^{m\times p}$ and upper slope $K_2\in \mathbb R^{m\times p}$, are described by 
\begin{align}\label{eq:condCircle} 
(-K_1\zeta + \tilde a(\zeta))^\top (K_2 \zeta- \tilde a(\zeta)) \geq 0, 
\end{align}
{see Fig.~\ref{fig:sector_k1k2} for the scalar case with $K_1=k_1<k_2=K_2$}. In terms of a quadratic form, (\ref{eq:condCircle}) becomes
\begin{align}
 w(\zeta,\tilde a(\zeta))&\geq 0 \quad \text{with } 
\nonumber \\
 w(\zeta,\alpha)&= \begin{bmatrix} \zeta \\ \alpha \end{bmatrix}^\top
 \underbrace{\mathrm{sym}\left(\begin{bmatrix} - K_1^\top K_2 & K_1^\top \\
 K_2 & -I
\end{bmatrix} 
\right) }_{\Pi=\Pi^\top}
 \begin{bmatrix} \zeta \\ \alpha \end{bmatrix}. 
\label{eq:wCircle} 
\\[-3em]\nonumber
\end{align}
\absatz{{Choice of the sector}}
{The sector type should be chosen that fits best to $\tilde a(\zeta)$, at least locally around $\zeta=0$. The parameter in the matrix $\Pi$, i.e., $\gamma>0$ in (\ref{eq:wLinearNormBound}), $1/\rho>0$ in (\ref{eq:wOutputPassivity}), or, typically, 
 $k_1\in \mathbb R$ in $K_1=k_1 I_m$ with $m=p$ in (\ref{eq:wCircle}), is left as a degree of freedom. Its admissible range depends on the nominal system's robustness under the perturbation structure~$B,C$. 

For instance, for the function $\tilde a$ in Fig.~\ref{fig:sector_rho}, the sector $[0, 1/\rho]$ is better suited than $[-\gamma,\gamma]$ used in Fig.~\ref{fig:sector_gamma}. The latter does not distinguish between $\tilde a$ and $-\tilde a$, which is why the maximum admissible $\gamma$ might be much smaller than $1/\rho$. 
For a saturation, shown in Fig.~\ref{fig:sector_k1k2}, we choose a fixed upper slope $k_2$ close to the slope of $\tilde a$ in the origin, and are interested in the smallest admissible $k_1$. }

{The presented approach will yield admissible upper bounds on $\gamma$, $1/ \rho$, and $-k_1$ in Sec.~\ref{sec:FrequencyDomainResults}.}

\section{The underlying idea} \label{sec:idea}
Our objective is to derive a {continuous}, time-invariant Lyapunov-like function $V\colon {\mathbb R^n\to \mathbb R;}\; x\mapsto V(x)$. 
{Its defining inequality condition} will be expressed in terms of $D_{f(t,\cdot)}^+ V(x)$, which denotes the {upper right} Dini derivative of $V$ along solutions of $\dot x=f(t,x)$. If $V$ is continuously differentiable (which will be the case in Sec.~\ref{sec:ApproachesV}), {{the latter} simplifies to the classical directional derivative}
\begin{align*}
{D_{f(t,\cdot)}^+ V(x)=(\nabla V(x))^\top f(t,x).}
\end{align*} 

\absatz{{The construction principle for $V$}}
{The guiding idea is as follows.} We seek a function $V$ for which the derivative is less than or equal to $-w$, because then the derivative is at most zero, wherever $w$ is at least zero---and the latter characterizes the sector. 
More precisely: if a continuous function $V$ can be found for which the derivative along solutions of $\dot x=Ax-Ba(t,Cx)$ satisfies 
\begin{align}\label{eq:DVw}
&\forall x\in \mathbb R^n, \forall t\geq t_0:\nonumber \\
& \qquad\qquad D_{(A\,\cdot\,-Ba(t,C\cdot))}^+ V(x) \leq -w(Cx,a(t,Cx)),
\end{align}
then, for any $x$ from the set 
\begin{align} \label{eq:wgeqzero}
\Omega_0=\{x\in \mathbb R^n: w(Cx,a(t,Cx)) \geq 0 \quad \text{for all }t\geq t_0\}, 
\\[-3em]\nonumber
\end{align}
 we immediately achieve that 
\begin{align}\label{eq:DVnonpos}
&\forall x\in \Omega_0, \forall t\geq t_0: \nonumber\\[-1em]
&\quad D_{(A\,\cdot\,-Ba(t,C\cdot))}^+ V(x) \stackrel{(\ref{eq:DVw})}\leq -w(Cx,a(t,Cx)) \stackrel{(\ref{eq:wgeqzero})}\leq 0. 
\end{align}

\begin{figure}
\footnotesize
\centering
 \subfloat[Sector {$[-\gamma,\gamma]$}\label{fig:sector_gamma} 
]{
\includegraphics[]{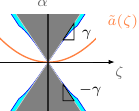}
}
\hfill
\subfloat[Sector {$[0,1/\rho]$}\label{fig:sector_rho} 
]{
\includegraphics[]{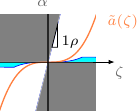}
}
\hfill
\subfloat[Sector {$[k_1,k_2]$}\label{fig:sector_k1k2} 
]{
\includegraphics[]{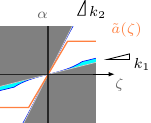}
}
\caption{ 
{The original sector is the non-gray region, cf.\ Fig.~\ref{fig:zerosuperlevelset}. The $\mathcal K$-reduced sector (plotted for $\ell$ from (\ref{eq:ell}) with $q=4$) is the turquoise marked reduction of this admissible region. Its relation to Fig.~\ref{fig:zerosuperlevelset} is explained in Fig.~\ref{fig:ell3dand2d}.} 
}
 \label{fig:reducedsector}
\end{figure}

\absatz{Some comments}
Concerning the construction of $V$ according to~(\ref{eq:DVw}): A function $V$ that satisfies (\ref{eq:DVw}) can easily be found, provided the chosen sector-describing quadratic form $w(\zeta,\alpha)$ from Sec.~\ref{sec:sectorDescription} is compatible with the system's robustness. Corresponding approaches will be discussed in Sec.~\ref{sec:ApproachesV}. In fact, these approaches do not only establish (\ref{eq:DVw}) for a special function $a(t,\zeta)$ of interest but (\ref{eq:DVw}) will hold simultaneously for arbitrary $a(t,\zeta)$. 

\begin{remark}[Relation to dissipativity]\label{rem:dissipativity}
Once (\ref{eq:DVw}) holds for any $u\in \mathbb R^m$ in place of $-a(t,\zeta)$, i.e., $D_{(A\,\cdot\,+Bu)}^+ V(x) \leq -w(Cx,-u)$, the latter can be interpreted as a QSR-dissipativity property of $(A,B,C)$ \cite{Willems.1972b}. Thus, as an alternative approach to the one in this article, it can be examined whether the combination (\ref{eq:ABaC}) with a sector-constrained nonlinearity is a neutral interconnection of dissipative elements \cite{Willems.1972b,Scherer.2022}. 
\fine\end{remark}

Concerning the set $\Omega_0$ from (\ref{eq:wgeqzero}): Note that $\Omega_0$ only depends on the sector defined by $w$ and on the nonlinearity $a$. It describes where the nonlinearity resides in the sector from Sec.~\ref{sec:sectorDescription} in terms of the points $x\in \mathbb R^n$ for which $\alpha=a(t,\zeta)$ with $\zeta=Cx$ satisfies the sector inequality $w(\zeta,\alpha)\geq 0$. 
If the graph of the nonlinear function $\zeta\mapsto a(t,\zeta)$ is globally within the sector (for any $t$), then \begin{align}
\Omega_0=\mathbb R^n.
\end{align} 
If the sector is chosen such that at least for small values of $\|\zeta\|$ the graph of $\zeta\mapsto a(t,\zeta)$ lies within the sector (for any $t$), then, by $\zeta=Cx$, the set $\Omega_0$ from (\ref{eq:wgeqzero}) also includes all states $x$ that exhibit a small enough norm value. Thus, 
\begin{align}
\Omega_0\subseteq \mathbb R^n 
\end{align} 
is indeed a domain around the 
zero equilibrium, as required when using (\ref{eq:DVnonpos}) in local Lyapunov arguments.
\begin{remark}[Relation to the S-procedure]
Let {$M_w$ and $Q=-M_d$} be symmetric matrices. {In} a simplified\footnote{{In fact, $z^\top Qz\geq 0$ holds for all $z$ for which $z^\top M_w z \geq 0$ if and only if $\exists\tau\geq 0: Q-\tau M_w\succeq 0$, assuming $\exists z_0: z_0^\top M_w z_0>0$ \cite{Boyd.1994}. This ``losslessness of the S-procedure'' \cite{Gusev.2006} has important optimality implications if a quadratic function is of interest that simultaneously yields a nonpositive derivative for all $(\zeta, a(t,\zeta))$ in the sector (note that with a free matrix $P$ in $V(x)=x^\top P x$, $\tau$ can be normalized).}} 
form, the S-procedure for LMIs states that $z^\top Q z\geq 0$ {(or $z^\top M_d z\leq 0$, cf.\ (\ref{eq:DVnonpos}))} holds for all $z$ for which $z^\top M_w z\geq 0$ if $Q-M_w\succeq 0$, \cite[Sec.~2.6.3]{Boyd.1994}. See also \cite{Gusev.2006,Yakubovich.2004,Aizerman.1965,Yakubovich.1973,Yakubovich.1971}. Thus, the underlying idea of (\ref{eq:DVw}) is {closely related to} the S-procedure. However, the further discussion will deviate from that concept. 
\fine\end{remark} 
Concerning the conclusion (\ref{eq:DVnonpos}): The semidefiniteness (\ref{eq:DVnonpos}) of the derivative {of $V$ along solutions} is not yet sufficient when attractivity of the zero equilibrium is to be addressed. {Thus, some strengthening is needed.}

\absatz{{Introduction of an offset function in the constraint}}
{A stronger conclusion than (\ref{eq:DVnonpos}) can be achieved in two ways: Either {by} replacing the defining inequality (\ref{eq:DVw}) by a more restrictive one (see the following remark) or {by} replacing the constraint (\ref{eq:wgeqzero}) by a more restrictive one. }
\begin{remark}[The alternative way: Strengthening the defining {inequality}] \label{rem:optionA}
In the construction of LMIs for asymptotic stability (derived from IQCs and the strict KYP lemma~\cite{Megretski.1997}, or derived via the strict S-procedure~\cite{Gusev.2006}) a quadratic term $e(x)=\varepsilon\|x\|_2^2$, with some small $\varepsilon>0$, is commonly introduced. In our setting, this corresponds to a replacement of the {defining inequality} (\ref{eq:DVw}) by \begin{align}\label{eq:DVwe}
D_{(A\,\cdot\,-Ba(t,C\cdot))}^+ V_e(x) \leq -w(Cx,a(t,Cx)) - e(x). 
\end{align} 
Then, the conclusion (\ref{eq:DVnonpos}) becomes negative definiteness
\begin{align}\label{eq:DVnonposE}
{\forall x\in \Omega_0, \forall t\geq t_0: \quad D_{(A\,\cdot\,-Ba(t,C\cdot))}^+ V_e(x) \leq - e(x)} 
\end{align} 
if $e(x)$ is chosen correspondingly. As a result, 
 proving asymptotic stability via $V_e$ 
becomes particularly straightforward. 
However, 
the question of existence of a function $V_e$ that satisfies (\ref{eq:DVwe}) depends on the respective choice of $e(x)$. In this article, we adopt a different approach that aims to keep the defining inequality for $V$ 
as simple as possible. 
\fine\end{remark}
{Instead of strengthening (\ref{eq:DVw}) to (\ref{eq:DVwe}), we strengthen the constraint $w(\zeta,\alpha)\geq 0$ to $w(\zeta,\alpha)\geq \ell(\zeta)$ with some \textit{offset function} $\ell:\mathbb R^p\to \mathbb R$. That is, we replace (\ref{eq:wgeqzero}) by}
\begin{align}\label{eq:OmegaEll}
\Omega_\ell = \{x\in \mathbb R^n: w(Cx,a(t,Cx)) \geq \ell(Cx) \text{ for all } t\geq t_0\}, 
\\[-3em]\nonumber 
\end{align} 
{
and the conclusion (\ref{eq:DVnonpos}) by the \textit{inequality chain} }
\begin{align}\label{eq:DVnonposL}
&{\forall x\in \Omega_\ell, \forall t\geq t_0:} \\*[-1em]
& D_{(A\,\cdot\,-Ba(t,C\cdot))}^+ V(x) \stackrel{(\ref{eq:DVw})}\leq -w(Cx,a(t,Cx)) \stackrel{(\ref{eq:OmegaEll})}\leq -\ell(Cx). \nonumber
\end{align} 
{Importantly, this strengthened conclusion still relies on the same $V$ satisfying the original defining inequality (\ref{eq:DVw}).} As offset function, we choose 
 an unspecified {class-$\mathcal K$} function 
\begin{align}\label{eq:ellK}
\ell(\zeta)=\kappa(\|\zeta\|), \quad \kappa\in \mathcal K . 
\end{align}

{Unless $C$ has rank $n$}, the inequality chain (\ref{eq:DVnonposL}) with (\ref{eq:ellK}) is still not a negative definiteness statement for the derivative of $V$. It only shows its partial negative definiteness w.r.t.\ $\zeta=Cx$ {instead of the full state $x$}. Nevertheless, {even if} {$\mathrm{rank}(C)\neq n$}, 
$V$ from~(\ref{eq:DVw}), which is referred to as \textit{Lyapunov-like function} in this article, will become useful for proving attractivity.

\absatz{{The $\mathcal K$-reduced sector}}
{For any {class-$\mathcal K$} offset function $\ell=\kappa(\|\cdot\|)$ from (\ref{eq:ellK}), the {above-described} strengthening of (\ref{eq:sectorConditionw}) to
\begin{align} \label{eq:wgeqell}
w(\zeta,\alpha)&\geq \ell(\zeta) 
\end{align}
results in a subset of the original sector, see Fig.~\ref{fig:reducedsector} for an example. 
If we simply choose a quadratic offset function
\begin{align}
\ell(\zeta)= \varepsilon \|\zeta\|^2, \quad \varepsilon>0 \label{eq:ellQuad}
\end{align}
{(turquoise line in Fig.~\ref{fig:ell} and \ref{fig:ell2}), then} inequality (\ref{eq:wgeqell}) 
{can be rewritten} as $\tilde w(\zeta,\alpha):=w(\zeta,\alpha)-\varepsilon \zeta^\top \zeta \geq 0$. Relying on a quadratic $\tilde w$, the latter still describes a cone. 
However, according to (\ref{eq:ellK}), $\ell(\zeta)$ can also be some nonquadratic {class-$\mathcal K$} function in $\|\zeta\|$ like} 
\begin{align}\label{eq:ell}
\ell(\zeta)= h\left\{\begin{array}{lr@{}l}(\tfrac 1 c \|\zeta\|)^q, & \|\zeta\|&<c \\ \tfrac 1 c \|\zeta\|, & \|\zeta\|&\geq c,
\end{array}\right. 
\end{align}
with $q\geq 2,h>0,c>0$ ({dark line} in Fig.~\ref{fig:ell}, {where} $q=4$). With such a nonquadratic offset function{,} the {\textit{{$\mathcal K$-reduced} sector}, i.e., the} set of $(\zeta,\alpha)$ pairs that satisfy (\ref{eq:wgeqell}), 
is no longer a cone, see {Fig.~\ref{fig:boundary} (top view on Fig.~\ref{fig:ellnonquad}) and} Fig.~\ref{fig:reducedsector}.

{As an advantage over quadratic offset functions (\ref{eq:ellQuad}), nonquadratic ones like (\ref{eq:ell}) } can achieve that the strengthened constraint (\ref{eq:wgeqell}) even holds for functions $a(t,\zeta)=\tilde a(\zeta)$ that, at $\zeta=0$, are tangential to the original sector bound, see Fig.~\ref{fig:sector_rho}. Moreover, in terms of global results, $\tilde a(\zeta)$ may asymptotically ($\|\zeta\|\to \infty$) become parallel to the sector bound (with a nonzero distance). {For such cases the proposed approach will still prove attractivity or asymptotic stability. In contrast, the usual small-gain theorem would suffer from the fact that such functions~$\tilde a$ are, in terms of their gain, not distinguishable from a linear perturbation $\tilde a(\zeta)=K\zeta$ on the original sector boundary. For the latter, however, only (\ref{eq:DVnonpos}), not yet (\ref{eq:DVnonposL}) is valid. }

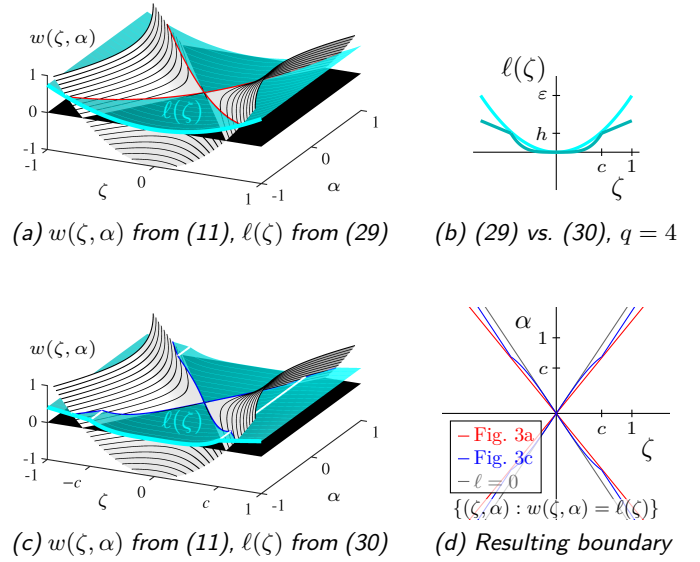
\begin{figure}
\centering 
\subfloat[{$w(\zeta,\alpha)$ from (\ref{eq:wLinearNormBound}), {$\ell(\zeta)$} from (\ref{eq:ellQuad})} \label{fig:ell2}
]{ 
\includegraphics[]{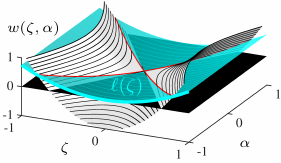}
}
\hfill
\subfloat[(\ref{eq:ellQuad}) vs.\ (\ref{eq:ell}){, $q=4$}\label{fig:ell}
]{ 
\includegraphics[]{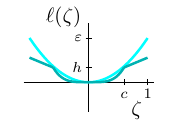}
}
\\
\subfloat[{$w(\zeta,\alpha)$ from (\ref{eq:wLinearNormBound}), {$\ell(\zeta)$} from~(\ref{eq:ell})}
\label{fig:ellnonquad}
]{ 
\includegraphics[]{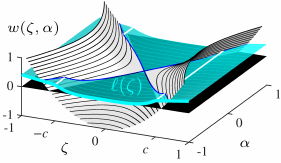}
}
\hfill
\subfloat[{Resulting boundary}\label{fig:boundary}]{ 

\begin{tikzpicture}
\begin{scope}

\path[clip] (-1.5,-1.4) rectangle (1.5,1.4);

\node[inner sep=0pt,anchor=center] (fig_pdf) 
at (0,0) 
{\includegraphics[]{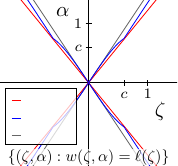}};

\definecolor{col}{rgb}{1,0,0}
\node[anchor=west,color=col,inner sep=2pt,scale=0.75] at (-1.15,-0.3) {Fig.~\ref{fig:ell2}};

\definecolor{col}{rgb}{0, 0, 1}
\node[anchor=west,color=col,inner sep=2pt,scale=0.75] at (-1.15,-0.6) {Fig.~\ref{fig:ellnonquad}};

\definecolor{col}{rgb}{0.4,0.4,0.4}
\node[anchor=west,color=col,inner sep=2pt,scale=0.75] at (-1.15,-0.9) {$\ell=0$};

\end{scope}
\end{tikzpicture}
}

\caption{
{The {$\mathcal K$-reduced} sector is $\{(\zeta,\alpha)\in \mathbb R^2: w(\zeta,\alpha)\geq \ell(\zeta)\}$ with $\ell(\zeta)=\kappa(\vert\zeta\vert)$ being 
some {class-$\mathcal K$} function in $\vert\zeta\vert$.}
}
\label{fig:ell3dand2d}
\end{figure}

\begin{remark}[{Not simply a strict inequality}]
\label{rem:strictsector}
Unless {$\zeta$ is restricted to a bounded domain,} 
the {above introduced} existence of a {class-$\mathcal K$} function $\kappa$, such that $w(\zeta,\alpha)\geq \kappa(\|\zeta\|)$, is a stronger requirement than a strict inequality ${w(\zeta,\tilde a(\zeta))>0}$ {for $\zeta\neq 0$}. The latter is too weak as it would also be satisfied by functions that become tangential to the sector bound as $\|\zeta\|\to \infty$.
\fine\end{remark}

\section{{Main result}} 
\label{sec:attractivityConclusion}
Before stating the main result 
in Thm.~\ref{thm:ConvergenceStatement} below, we discuss some nonrestrictive assumptions. 
\absatz{Assumptions} \\*[-3em]
\begin{assumption}[Knowledge about a stabilizing linear perturbation] \label{asm:stabilizing} 
For at least one $K_{\mathrm{stab}}\in \mathbb R^{m\times p}$, for which $\tilde a(\zeta)=K_{\mathrm{stab}}\zeta$ belongs to the interior of the sector, i.e., 
\begin{align} \label{eq:kKstab}
\exists k>0: w(\zeta,K_{\mathrm{stab}}\zeta)\geq k \|\zeta\|_2^2 \text{ for all } \zeta\in \mathbb R^p,
\end{align}
 it is known that 
\begin{align} \label{eq:Kstab}
 A-BK_{\mathrm{stab}}C \quad \text{ is Hurwitz}. 
\\[-2.5em]\nonumber
\end{align}
\fine \end{assumption}
Assumption~\ref{asm:stabilizing} is not restrictive since, for any linear perturbation $K\zeta$ from the interior of the sector, 
\begin{align} \label{eq:stabilized}
\dot x= Ax -B\tilde a(Cx) \stackrel{\tilde a(\zeta)=K\zeta}= (A-BKC) x 
\end{align}
{is among the systems for which attractivity of the zero equilibrium is to be established.} 
If $\tilde a(\zeta)\equiv 0_m$ lies in the interior of the sector, which is true whenever $\Pi_{\zeta\zeta}\succ 0$ in (\ref{eq:wzetaalphaPi}), then we can choose $K_{\mathrm{stab}}=0_{m\times p}$. As a consequence, in this case, Assumption \ref{asm:stabilizing} only means that $A$ must be Hurwitz. In contrast, if $\tilde a(\zeta)\equiv 0_m$ is not in the interior of the sector, then $A$ is not required to be Hurwitz.

\begin{assumption}[Generalized validity of (\ref{eq:DVw})]\label{asm:DVlin}
The defining {inequality} (\ref{eq:DVw}) from Sec.~\ref{sec:idea} holds not only for the specific nonlinearity of interest but also for $a(t,\zeta)=K_{\mathrm{stab}}\zeta$ with $K_{\mathrm{stab}}$ from Assumption~{\ref{asm:stabilizing}.} 
\fine\end{assumption}
This assumption is clearly satisfied if 
(\ref{eq:DVw}) holds for any arbitrary function $a(t,\zeta)$ (cf.\ Rem.~\ref{rem:dissipativity}), which will 
be the case in the approaches discussed in the next section. 

\begin{assumption}[Continuity and zero property of $V$] \label{asm:V0conv}
$V\colon \mathbb R^n\to \mathbb R$ is continuous, and $V(0_n)=0$. 
\fine\end{assumption}
Assumption~\ref{asm:V0conv} is clearly satisfied by quadratic forms $V(x)=x^\top P x$, $P\in \mathbb R^{n\times n}$. The latter will be used as ansatz for $V$ in the next section. 

\absatz{{Attractivity and asymptotic stability theorems}} 
Depending on whether the core perturbation function $a(t,\zeta)$ {from (\ref{eq:ABaC})} resides globally {(in the sense of for any $\zeta\in \mathbb R^p$)} or only locally within the {$\mathcal K$-reduced} sector, we obtain a global or local attractivity result for the zero equilibrium. The local one comes with a simple estimate of the domain of attraction. 
{To this end, we denote sublevel sets $L_{\leq\level}$ and strict sublevel sets $L_{<\level}$ of $V$ with level $\level>0$ by}
\begin{align}\label{eq:sublevelset}
L_{\leq \level}&=\{x\in \mathbb R^n: V(x)\leq \level\}, 
\\
 L_{<\level}&=\{x\in \mathbb R^n: V(x) <\level\}. \label{eq:strictSublevelset}
\end{align}

\begin{theorem}[Attractivity]\label{thm:ConvergenceStatement}
{Consider the perturbed system (\ref{eq:ABaC}).} Suppose Assumption \ref{asm:stabilizing} holds, by which $K_{\mathrm{stab}}$ exists. 
{Let $V$ be constructed according to~(\ref{eq:DVw}), i.e., for all $x\in \mathbb R^n$ and $t\geq t_0$ 
\begin{align*} 
D_{(A\,\cdot\,-Ba(t,C\cdot))}^+ V(x) \leq -w(Cx,a(t,Cx)), 
\end{align*}
satisfying Assumption~\ref{asm:DVlin} and \ref{asm:V0conv}. 
}
Consider $\Omega_{\kappa(\|\cdot\|)}\subseteq \mathbb R^n$ {from~(\ref{eq:OmegaEll})} with 
$\ell(\zeta)=\kappa(\|\zeta\|)$ for some $\kappa\in \mathcal K$. 
\begin{enumerate}[label=(\roman*)]
\item If $\Omega_{\kappa(\|\cdot\|)}=\mathbb R^n$, i.e., the perturbation $a(t,\zeta)$ is for all $\zeta\in \mathbb R^p$ and all $t\geq t_0$ in the {$\mathcal K$-reduced} sector, 
then 
\begin{align*}
\forall x_0\in \mathbb R^n: \quad \|x(t;t_0,x_0)\|\to 0 \text{ as } t\to \infty, 
\end{align*}
i.e., the zero equilibrium {of (\ref{eq:ABaC})} is globally attractive.
\item \label{item:locallyStab} 
Otherwise, 
let {$\level> 0$} be a level for which the sublevel set (\ref{eq:sublevelset}) satisfies {$L_{\leq \level} \subseteq \Omega_{\kappa(\|\cdot\|)}$}. Then, 
\begin{align*}
{x_0\in L_{<\level}} \quad \Longrightarrow \quad \|x(t;t_0,x_0)\|\to 0 \text{ as } t\to \infty, 
\end{align*}
i.e., $L_{<\level}$ is a subset of the domain of attraction {of the zero equilibrium of (\ref{eq:ABaC})}. 
\\[-2.75em]
\end{enumerate}
\fine\end{theorem}
\begin{proof}
 We denote by $\mathcal U$ the nullspace of the observability Gramian (equivalently, the unobservable subspace of $(A,C)$). Henceforth, for any $x\in \mathbb R^n=\mathcal U^\bot \oplus \mathcal U$,
\begin{align}\label{eq:xBot}
x=x_\bot+x_\| \quad \text{with} \quad x_\bot\in \mathcal U^\bot \text{ and } x_\| \in \mathcal U 
\end{align}
 refers to the decomposition of $x$ into $x_\bot$ in the observable and $x_\|$ in the unobservable subspace of $(A,C)$. 
The proof relies on {lemmas} that are presented in Appendix \ref{sec:ProofConv}.

\begin{enumerate}
\item \label{it:VLowerBoundObsv} 
Lemma~\ref{lem:VLowerBoundObsv} {proves} that $V(x)$ is partially\footnote{We use the description ''partially w.r.t. \ldots'' {to mark} that a property is only true on a special, known subspace.} positive definite and partially radially unbounded w.r.t.\ the $(A,C)$-observable part $x_\bot$ of the argument $x$ (defined in (\ref{eq:xBot})). 
To this end, {we exploit} the inequality chain (\ref{eq:DVnonposL}) 
for the derivative of $V$ along solutions of the stabilized system {(\ref{eq:stabilized}). Therefore,} 
$a(t,\zeta)=K_{\mathrm{stab}} \zeta$ from (\ref{eq:Kstab}) and $\ell(\zeta)=k\|\zeta\|_2^2$ from~(\ref{eq:kKstab}) { is used in (\ref{eq:DVKstabEllQuad}).}

\item \label{it:sublevelset} 
Lemma~\ref{lem:sublevelset} shows that any {strict} sublevel set {$L_{<\level}\subset L_{\leq \level} \subseteq \Omega_{0}$} is a positively invariant set of (\ref{eq:ABaC}). {To this end, we exploit the inequality chain (\ref{eq:DVnonposL}) for the derivative of $V$ along local solutions of (\ref{eq:ABaC}), which shows that this derivative is nonpositive.}

\newcounter{CounterConv}
\setcounter{CounterConv}{\value{enumi}} 
\end{enumerate}
{Note that $V$ is only proven to be partially positive definite in~(\ref{eq:partialPosDef}) and positively invariant sublevel sets might become unbounded. {Consequently}, classical arguments do not yet establish boundedness of {solutions $x(\cdot)$ of (\ref{eq:ABaC})}.}
\begin{enumerate} \setcounter{enumi}{\value{CounterConv}}

\item \label{it:boundednessBot} Lemma~\ref{lem:xBotBounded} shows that the $(A,C)$-observable projection~$x_\bot(\cdot)$ of solutions $x(\cdot)$ with {$x(t_0)\in L_{<\level}$} 
is bounded, using item~\ref{it:VLowerBoundObsv} and~\ref{it:sublevelset}.

\item \label{it:boundedness} Lemma~\ref{lem:boundedness} shows that also $x(\cdot)$ is bounded, using the \textit{bounded-input-bounded-state (BIBS)} property of $\dot x=(A-BK_{\mathrm{stab}}C)x+Bu$ and Assumption~\ref{asm:aLocLip}. 
\setcounter{CounterConv}{\value{enumi}} 
\end{enumerate}
{{Having boundedness of local solutions $x(\cdot)$ from (\ref{eq:ABaC}), we apply Barbalat's lemma to prove that $\|Cx(t)\|\to 0$ as $t\to\infty$. Note that, for the considered system class, maximal solutions $x(t)$ that are bounded are also guaranteed to exist on $[t_0,\infty)$ (forward completeness).}}
\begin{enumerate} \setcounter{enumi}{\value{CounterConv}}
\item\label{it:uniformlyContComposition} Lemma~\ref{lem:uniformlyContComposition} {states} that $b_\kappa\colon [t_0,\infty)\to \mathbb R,$
\begin{align}\label{eq:bkdef}
 t\mapsto b_\kappa(t):= \kappa(\|C x(t)\|),
\end{align}
is uniformly continuous {relying on} the boundedness of $x(\cdot)$ from item~\ref{it:boundedness} and the local Lipschitz property of $a$ from Assumption~\ref{asm:aLocLip}.

\item\label{it:Bbounded} Lemma~\ref{lem:Bbounded} shows that, given an initial value {$x(t_0)\in L_{<\level}\subset L_{\leq \level} \subseteq \Omega_{\kappa(\|\cdot\|)}$,} the resulting $B_\kappa\colon [t_0,\infty){\to \mathbb R},$
\begin{align}
 t\mapsto B_\kappa(t):= \int_{t_0}^t \kappa(\|C x(\tau )\|)\,\mathrm d \tau,
\end{align}
is bounded, using the positive invariance of {$L_{<\level}$} 
from item~\ref{it:sublevelset}, the inequality chain (\ref{eq:DVnonposL}) and the nonnegativity of $V$ from item~\ref{it:VLowerBoundObsv}.

\item\label{it:BLimit} Note that $B_\kappa(\cdot)$ has a nonnegative derivative $\kappa(\|C x(\cdot)\|) \geq 0$, and thus $B_\kappa(\cdot)$ is nondecreasing. Combined with the boundedness of $B_\kappa(\cdot)$ from item~\ref{it:Bbounded}, we conclude that the limit $\lim_{t\to\infty}B_\kappa(t)$ exists and is finite. 

\item \label{it:Barbalat}Using the existence of $\lim_{t\to\infty}B_\kappa(t)$ from item~\ref{it:BLimit}, and using the uniform continuity of its derivative $b_\kappa(\cdot)$ from item~\ref{it:uniformlyContComposition}, we exploit \textit{Barbalat's lemma} (Lemma~\ref{lem:Barbalat}) to conclude that $\lim_{t\to\infty} b_\kappa(t)=0$. 
As a result, $Cx(t)$ must converge to zero. 
\setcounter{CounterConv}{\value{enumi}} 
\end{enumerate}
{As a final step, we show that $\|x(t)\|\to 0$ as $t\to\infty$.}
\begin{enumerate} \setcounter{enumi}{\value{CounterConv}}

\item Lemma~\ref{lem:convOfx} proves that 
$x(t)$ converges to zero, using the \textit{converging-input-converging-state {(0-CICS)}} property of $\dot x={(A-BK_{\mathrm{stab}}C)x}+Bu$ and Assumption~\ref{asm:aLocLip}, by which ${\zeta(t)\to 0_p}$ implies $a(t,\zeta(t))\to a(t,0_p)=0_m$.\vspace{-2em}
\end{enumerate}\nopagebreak
\end{proof}
 
The above statement relies on $V$ being only a Lyapunov-like function. {Already if $(A,C)$ is observable, item~\ref{it:VLowerBoundObsv} establishes that $V$ is positive definite, while (\ref{eq:DVnonpos}) ensures that its derivative along {local} solutions of (\ref{eq:ABaC}) is nonpositive. Consequently, $V$ becomes a Lyapunov function that {proves} uniform stability. Combined with the above attractivity, asymptotic (though, in the time-varying case, not yet uniform asymptotic) stability can be concluded. } 

If $C$ is even chosen with {$\mathrm{rank}(C)=n$}, 
then the above proof simplifies considerably as $V$ even becomes a classical Lyapunov function for uniform asymptotic stability (i.e., uniform attractivity and uniform stability). Notably, the effect of adapting $C$ to {have {$\mathrm{rank}(C)=n$} 
(cf.~Example~\ref{exmp:BaC})} is the same as if {the strengthened defining inequality} 
(\ref{eq:DVwe}) with some positive definite $e(x)$ 
 would have been chosen. 
\begin{theorem}[{Uniform asymptotic stability}]\label{thm:ConvergenceFullRankC} {Let the decomposition (\ref{eq:ga}) be chosen such that} {$\mathrm{rank}(C)=n$.} 
If the conditions of Thm.~\ref{thm:ConvergenceStatement} {hold}, 
{then} not only global or local attractivity, but even global or local uniform asymptotic stability {w.r.t.\ all initial times $\tilde t_0\geq t_0$} is proven. 
\fine\end{theorem}
\begin{proof}
The function $V$ satisfies the conditions of the classical Lyapunov theorem, {see \cite[Thm.~4.9]{Khalil.2002} if $V\in C^1$, or, more generally, \cite[Thm.~8.3 and 11.5]{Yoshizawa.1966} for a Dini-derivative-based formulation}: 
\begin{enumerate}
\item Positive definiteness and radial unboundedness of $x\mapsto V(x)$ is satisfied due to item~\ref{it:VLowerBoundObsv} in the above proof, once $(A,C)$ is observable, which is clearly the case if {$\mathrm{rank}(C)=n$}. 
To be more precise, it holds 
\begin{align*} 
\exists k_{1}>0,\forall x\in \mathbb R^n: \quad k_1\|x\|^2 \leq V(x) 
\end{align*}
by (\ref{eq:partialPosDef}), where $x_\bot=x$. 
\item Negative definiteness of the derivative of $V$ is met by
\begin{align*} 
&{\forall x\in \Omega_{\kappa(\|\cdot\|)}, \forall t\geq t_0: }
\\
&
\quad D_{(A\,\cdot\,-Ba(t,C\cdot))}^+ V(x)\leq -\tilde \kappa(\|x\|)
\end{align*}
 with some $\tilde \kappa\in \mathcal K$ 
since, if {$\mathrm{rank}(C)=n$}, 
the estimate
 $\|Cx\|_2^2=x^\top C^\top C x \geq k 
\|x\|_2^2$, 
{${k:=\lambda_{\min}(C^\top C)>0}$,} 
\begin{align*}
\ell(Cx)&=\kappa(\|C x\|_2) \geq 
\kappa\big(\sqrt{k} \|x\|_2\big)
=:\tilde \kappa(\|x\|_2) 
\end{align*}
{applies to (\ref{eq:DVnonposL}). }
\\[-2.75em]
\end{enumerate}
 \end{proof}

\section{Quadratic solutions} \label{sec:ApproachesV}
In order to find a function $V$ that satisfies the defining inequality (\ref{eq:DVw}) from Sec.~\ref{sec:idea}, a quadratic ansatz
\begin{align} \label{eq:ansatzV}
V(x)=x^\top P x, \quad P=P^\top \in \mathbb R^{n\times n}
\end{align} 
is used throughout the following approaches. 

Although the resulting conditions recover well-known concepts, 
the derivations are provided to make their relation to the proposed framework explicit. 

\absatz{Approach A: Linear matrix inequality (LMI)}
For the {quadratic} ansatz (\ref{eq:ansatzV}), the derivative appearing in {the defining inequality} (\ref{eq:DVw}) 
can be written as 
\begin{align}\label{eq:dVxPAx}
& D_{(A\,\cdot\,-Ba(t,C\cdot))}^+ V(x) 
= 2 x^\top P (Ax-Ba(t,Cx))
\\
&\quad = \begin{bmatrix} x \\ -a(t,Cx) \end{bmatrix}^\top
\underbrace{\begin{bmatrix} 2PA & 2 PB \\ 0 & 0 \end{bmatrix}}_{=:M_{d}(P)} \begin{bmatrix} x \\ -a(t,Cx) \end{bmatrix}
\nonumber \\
&\quad=
\begin{bmatrix} x \\- a(t,Cx) \end{bmatrix}^\top
\underbrace{\begin{bmatrix} PA+A^\top P & PB \\ B^\top P & 0 \end{bmatrix}}_{\mathrm{sym}( M_{d}(P))} \begin{bmatrix} x \\ -a(t,Cx) \end{bmatrix}\nonumber 
\label{eq:MdP}
\\[-1.5em]
\end{align}
with\footnote{A symmetrization {(\ref{eq:sym})} is applied since definiteness properties, in the present paper, are only defined for Hermitian matrices. } an in general indefinite\footnote{If $PB\neq 0$ the matrix in (\ref{eq:MdP}) cannot be semidefinite since then, {by Lemma~\ref{lem:semidefDecomp}, it would be writable as} $\pm Z^\top Z$, where $Z_2=0$ in $Z=[Z_1,Z_2]$ contradicts $Z_1^\top Z_2=PB\neq 0$.
} matrix $\mathrm{sym}( M_{d}(P))$. 
Similarly, $w(Cx,a(t,Cx))$ from (\ref{eq:wzetaalphaPi}) can be written as 
\begin{align}
&w(Cx,a(t,Cx)) \label{eq:Mw}
\\*
&=\begin{bmatrix} x \\ -a(t,Cx) \end{bmatrix}^\top
\underbrace{\begin{bmatrix} C^\top \Pi_{\zeta\zeta}C & -C^\top \Pi_{\zeta \alpha}\\ -\Pi_{\zeta \alpha}^\top C & \Pi_{\alpha\alpha} \end{bmatrix}}_{=:M_w} \begin{bmatrix} x \\ -a(t,Cx) \end{bmatrix}, 
 \nonumber
\end{align}
where the matrix $M_w$ is {typically} indefinite due to the underlying matrix $\Pi$ from Sec.~\ref{sec:sectorDescription}. 
In terms of these matrices, {the defining inequality} (\ref{eq:DVw}) becomes
\begin{align} 
&\forall x\in \mathbb R^n, {\forall t\geq t_0}: \label{eq:DVwMatrix} \\
& \begin{bmatrix} x \\ -a(t,Cx) \end{bmatrix}^\top
\Big( \mathrm{sym}\big(M_d(P)\big) + M_w\Big) \begin{bmatrix} x \\ -a(t,Cx) \end{bmatrix} 
\leq 0, \nonumber
\end{align}
which immediately leads to the LMI in (\ref{eq:LMI}) {below.}
\begin{theorem}[LMI]\label{thm:LMI}
Suppose $P=P^\top\in \mathbb R^{n\times n }$ satisfies 
\begin{align}\label{eq:LMI}
 \underbrace{\mathrm{sym}( M_{d}(P)) + M_w}_{=:M(P)}\preceq 0_{(n+m)\times (n+m)}, 
\end{align}
with $\mathrm{sym}( M_{d}(P))$ from (\ref{eq:MdP}) and $M_w$ from (\ref{eq:Mw}). Then $V(x)=x^\top P x$ satisfies (\ref{eq:DVw}) from Sec.~\ref{sec:idea} for any perturbation function $a(t,\zeta)$. 
\fine\end{theorem} 
 {A corresponding matrix $P$ can be found via a semidefinite programming solver.} 
\absatz{Approach B: Algebraic Riccati equation (ARE)}
An alternative (but strongly related, see Sec.~\ref{sec:solvabilityEquivalence}) approach to find an appropriate matrix $P$ in (\ref{eq:ansatzV}) relies on an~ARE. 
The following assumption ensures that $(-\Pi_{\alpha\alpha})^{-1/2}$, needed below, exists and is real.
\begin{assumption} \label{asm:PIaaNeg}
In (\ref{eq:wzetaalphaPi}), let $\Pi_{\alpha\alpha}\prec 0_{m\times m}$. 
\fine\end{assumption}
{Thus, in contrast to the LMI discussed above, {the} subsequent} 
{considerations do not include the limit case $\Pi_{\alpha\alpha}= 0_{m\times m}$.} 
\begin{theorem}[ARE]\label{thm:ARE}
Consider the abbreviation $\tilde B(P):=(PB - C^\top \Pi_{\zeta\alpha})(-\Pi_{\alpha\alpha})^{-1/2}$. 
Suppose $P=P^\top\in \mathbb R^{n\times n }$ satisfies the ARE 
\begin{align}\label{eq:ARE}
 &PA+A^\top P = - C^\top \Pi_{\zeta\zeta} C - \tilde B(P)\, \tilde B^\top\!(P). 
\end{align}
Then $V(x)=x^\top P x$ satisfies (\ref{eq:DVw}) for any $a(t,\zeta)$.
\fine\end{theorem}
\begin{proof}
The described {function $V$} 
has the derivative 
\begin{align*}
&D_{(A\,\cdot\,-Ba(t,C\cdot))}^+ V(x)\\
&= x^\top (P A+A^\top P)x-2 x^\top P Ba(t,Cx) 
\\
&\stackrel{\!\!(\ref{eq:ARE})\!\!}= x^\top \!\big(- C^\top\! \Pi_{\zeta\zeta} C - \tilde B(P)\,\tilde B^\top\!(P)\big)x-2 x^\top \!P Ba(t,Cx). 
\end{align*}
Adding $0=-\|\tilde b+ \tilde a\|_2^2+\tilde b^\top \tilde b +2 \tilde b^\top \tilde a + \tilde a^\top \tilde a$ where $\tilde b^\top:=x^\top \tilde B (P)$ and $\tilde a:= (-\Pi_{\alpha\alpha})^{1/2} a(t,Cx)$, i.e., {where}
\begin{align*}
\tilde b^\top \tilde b &= x^\top \tilde B(P)\, \tilde B^\top\!(P) x, \\*
2\tilde b^\top \tilde a & = 2 x^\top (PB - C^\top \Pi_{\zeta\alpha}) a(t,Cx), 
\\*
\tilde a^\top\tilde a &= a^\top (t,Cx) (-\Pi_{\alpha\alpha}) a(t,Cx),
\end{align*}
 gives 
\begin{align*}
&D_{(A\,\cdot\,-Ba(t,C\cdot))}^+ V(x)\\
&\quad = - x^\top C^\top \Pi_{\zeta\zeta} C x-2 x^\top C^\top \Pi_{\zeta \alpha} a(t,Cx) 
\\*
&\quad \hspace{2cm}-a^\top\! (t,Cx) \Pi_{\alpha\alpha} a(t,Cx) -\|\tilde b +\tilde a\|_2^2
\\
&\quad= - w(Cx,a(t,Cx)) -\|\tilde b +\tilde a\|_2^2. 
\\[-3em]
\end{align*}
\end{proof} 
Standard implementations (e.g., \texttt{icare} in MATLAB) can be used to obtain {a} (nonunique) solution of the ARE. 
\absatz{Approach C: Matrix equation} 
In the field of passivity (i.e., $\Pi_{\zeta\zeta}=\Pi_{\alpha\alpha}=0, 2\Pi_{\zeta\alpha}=I$), 
 $P$ is typically described via a matrix equation. 
Such a matrix equation can be formulated for any of the considered sectors, which thus provides a further (but again strongly related) approach. 
\begin{theorem}[Matrix equation] \label{thm:ME}
Suppose, for some $L\in \mathbb R^{q\times n}$ and $W\in \mathbb R^{q\times m}$, $q\in \mathbb N$, the matrix $P=P^\top\in \mathbb R^{n\times n }$ satisfies the matrix equation 
\begin{subequations}\label{eq:matrixEquation}
\begin{align}
\underbrace{PA+A^\top P + C^\top \Pi_{\zeta\zeta} C}_{=: M_{11}(P)}&= -L^\top L, 
\\
\underbrace{PB - C^\top \Pi_{\zeta\alpha}}_{=:M_{12}(P)}&= -L^\top W, 
\\
\Pi_{\alpha\alpha} \,&= -W^\top W. \label{eq:PiaaWW}
\end{align}
\end{subequations}
Then $V(x)=x^\top P x$ satisfies (\ref{eq:DVw}) for any perturbation function $a(t,\zeta)$. 
\fine\end{theorem}
\begin{proof}
Observe that (\ref{eq:matrixEquation}) can be interpreted as block-componentwise evaluation of 
\begin{align} \label{eq:MLW}
\underbrace{\begin{bmatrix}
M_{11}(P) & M_{12}(P) \\
M_{12}^\top(P) & \Pi_{\alpha\alpha}
\end{bmatrix} }_{M(P)}
=
-\begin{bmatrix} L^\top \\ W^\top \end{bmatrix}
\begin{bmatrix} L & W \end{bmatrix}.
\end{align}
The left-hand side equals $\mathrm{sym}(M_d(P)) + M_w$ from (\ref{eq:MdP}) and (\ref{eq:Mw}). Recall that a matrix $N$ is positive semidefinite if and only if it can be decomposed as $N=Z^\top Z$ (see Lemma~\ref{lem:semidefDecomp} below). Thus, by $N=-M(P)$ and $Z=[L\;W]$, (\ref{eq:MLW}) {implies $M(P)\preceq 0$, and the assertion follows from Thm.~\ref{thm:LMI}.} 
\end{proof}

{\absatz{Advantages of each approach}
The main advantage of the LMI (\ref{eq:LMI}) is that a parameter optimization---in particular the search for the best possible sector describing parameter $\gamma$, $\rho$ or $k_1$ from Sec.~\ref{sec:sectorDescription} for which the LMI is feasible---can easily be included in the semidefinite programming problem. 

An advantage of the ARE (\ref{eq:ARE}) compared to an LMI is that an ARE can also be solved efficiently for very large systems. Such large-scale matrices arise, for example, from the discretization of infinite-dimensional problems \cite{Scholl.2024b}. Moreover, for very small problems like Example~\ref{ex:finalexample}, it can even be possible to find (parameter-dependent) analytical solutions, see \cite[Rem.~6.2.14]{Scholl.2024c}. 

Finally, the matrix equation is useful for recognizing the interrelations between the approaches, see below.} 

\section{Solvability equivalence}\label{sec:solvabilityEquivalence}
The aforementioned problems of an ARE, LMI, etc., are closely related, see, e.g., Willems' famous work \cite{Willems.1971c}. 
 
As already visible from (\ref{eq:MLW}), the matrix equation (\ref{eq:matrixEquation}) is equivalent to the LMI provided that the aggregated matrix $[\;L\;\;W\;]$ 
is allowed to be square.

Moreover, the LMI can be rewritten as an algebraic Riccati inequality (ARI), applying the following lemma to $N=-M(P)$ {from (\ref{eq:LMI})}, where $M_{22}(P)=\Pi_{\alpha\alpha}$.
\begin{lemma}[{\cite[Thm.~1.12]{Horn.2005}}]
Let $N=N^\top$ be a real matrix with a right-lower submatrix $N_{22}\succ 0$. Then 
\begin{align*}
N\succeq 0 \quad \Longleftrightarrow \quad N/N_{22}&\succeq 0, \\
\text{ where } N/N_{22}&=N_{11}-N_{12}N_{22}^{-1}N_{12}^\top 
\end{align*}
in terms of the Schur complement $N/N_{22}$ of $N$. 
\fine\end{lemma}
Both considerations immediately lead to the conclusion that the LMI, ARI, and matrix equation are equivalent. 
\begin{proposition} \label{prop:LMIequivalences}
The following are equivalent:
\begin{enumerate}[label=(\alph*1)]
\item\label{it:LMI} LMI (\ref{eq:LMI}), i.e., $M(P)\preceq 0$,
\item \label{it:ARI} algebraic Riccati inequality (ARI) {(provided Asm.~\ref{asm:PIaaNeg} holds)}
\begin{align} \label{eq:ARI}
\underbrace{PA+A^\top P + C^\top \Pi_{\zeta\zeta} C}_{M_{11}(P)} + \underbrace{\tilde B(P)\, \tilde B^\top\!(P)}_{\hspace{-2cm}M_{12}(P) (-\Pi_{\alpha\alpha})^{-1}M_{12}^\top(P)\hspace{-2cm}} \preceq 0 
\end{align}
with $\tilde B(P)$ defined in Thm.~\ref{thm:ARE}, 
\item \label{it:ME} matrix equation (\ref{eq:matrixEquation}) with $q=n+m$ (i.e., the aggregated matrix $\begin{bmatrix} L & W\end{bmatrix}$ is square). 
\end{enumerate} 
\end{proposition}

If, however, the aggregated matrix $[\;L\;\;W\;]$ 
from (\ref{eq:MLW}) is only rectangular since $W$ is supposed to be square, then the matrix equation becomes equivalent to the ARE (see (\ref{eq:M11Kres}) below). 

In terms of the LMI, such a row restriction on $[\;L\;\;W\;]$ comes along with a rank constraint, when the next lemma is applied to $N=-M(P)$ from (\ref{eq:MLW}). 
\begin{lemma}[{\cite[Ch.~5.5, Cor.~2]{Lancaster.1985}}]\label{lem:semidefDecomp}
There exists a decomposition $N=Z^\top Z$, where $Z$ has $q$ rows, if and only if $N\succeq 0$ and {$\mathrm{rank}(N)\leq q$}. 
\fine\end{lemma} 
Thus, the ARE is equivalent to constrained versions of the LMI and matrix equation. 
\begin{proposition}\label{prop:AREequivalences}
The following are equivalent:
\begin{enumerate}[label={(\alph*2)}]
\item \label{it:LMIc} LMI (\ref{eq:LMI}) with rank constraint 
{$\mathrm{rank}(M(P))\leq m$}, 
\item \label{it:ARE} ARE (\ref{eq:ARE}) (provided Asm.~\ref{asm:PIaaNeg} holds),
\item \label{it:MEc} matrix equation (\ref{eq:matrixEquation}) with $q=m$ (i.e., $W$ and $L$ have as many rows as $\Pi_{\alpha\alpha}$). 
\end{enumerate} 
\end{proposition}
\begin{proof}
It only remains to provide a reasoning on \ref{it:MEc} $\Leftrightarrow$ \ref{it:ARE}: 
By (\ref{eq:PiaaWW}) {and Asm.~\ref{asm:PIaaNeg}, $W^\top W=-\Pi_{\alpha\alpha}\succ 0$, and thus $W$ is invertible. With $K=W^{-1}L$, the right-hand side of (\ref{eq:MLW}) equals}
\begin{align}\label{eq:LtoK}
-\begin{bmatrix} L^\top \\ W^\top \end{bmatrix}
\begin{bmatrix} L & W \end{bmatrix}
= -\begin{bmatrix} K^\top \\ I_m \end{bmatrix}
(-\Pi_{\alpha\alpha})
\begin{bmatrix} K & I_m \end{bmatrix}{.}
\end{align}
A componentwise comparison of the left-hand side of (\ref{eq:MLW}) and the result of (\ref{eq:LtoK}) yields
\begin{align}
M_{11}(P)&=-K^\top(-\Pi_{\alpha\alpha})K, \label{eq:M11K}
\\*
M_{12}(P) &= - K^\top (-\Pi_{\alpha\alpha}) \label{eq:M12K}
\\
&\Rightarrow \quad M_{11}(P) = M_{12}(P) \,K. \label{eq:M11Kres}
\end{align}
{Note that} (\ref{eq:M11Kres}) with $K=-(-\Pi_{\alpha\alpha})^{-1} M_{12}^\top(P)$ from (\ref{eq:M12K}), is (\ref{eq:ARI}) as an equality, i.e., the ARE (\ref{eq:ARE}). 
\end{proof}

Due to the additional constraints, solvability of the problems in Proposition~\ref{prop:AREequivalences} implies solvability of the weaker problems in Proposition~\ref{prop:LMIequivalences}. However, according to the following lemma, the converse is true as well. 
\begin{lemma}[Solvability equivalence, see, e.g., {\cite{Scherer.1995}}]\label{lem:solvabilityEquivalence}
Assume $(A,B)$ is stabilizable. There exists a $P=P^\top\in \mathbb R^{n\times n}$ that solves 
 \ref{it:LMIc}, \ref{it:ARE}, or equivalently \ref{it:MEc} in Prop.~\ref{prop:AREequivalences}
 if and only if there exists a $P=P^\top\in \mathbb R^{n\times n}$ that solves \ref{it:LMI}, \ref{it:ARI}, or equivalently \ref{it:ME} in Prop.~\ref{prop:LMIequivalences}. 
\fine\end{lemma}
Note that stabilizability of $(A,B)$ is a weaker condition than the existence of $K_{\mathrm{stab}}$ from Assumption~\ref{asm:stabilizing}, and thus not restrictive in the present context. 

As a result, there exists a (nonunique) real symmetric solution $P$ in one of the above approaches if and only if there exists a solution in any of the other approaches. The question of solvability, which thus is simultaneously the feasibility of all the discussed approaches, is addressed via the KYP lemma in the frequency domain. 
 
\section{Origin of a {frequency-domain} condition}\label{sec:FDI}
Before {revisiting} the KYP lemma 
in the next section, this section motivates the frequency-domain condition on which the KYP lemma relies. 

Consider again the transition from {the defining inequality} (\ref{eq:DVw}) from Sec.~\ref{sec:idea} to the LMI (\ref{eq:LMI}). According to (\ref{eq:MdP}) and (\ref{eq:Mw}), the original requirement (\ref{eq:DVw}) with $V(x)=x^\top P x$ becomes
\begin{align}
&\forall x\in \mathbb R^n, u=-a(t,Cx): \nonumber \\*
& \underbrace{\begin{bmatrix} x \\ u \end{bmatrix} ^\top \mathrm{sym}( M_d(P)) 
\begin{bmatrix} x \\ u \end{bmatrix}}_{2x^\top P (Ax+Bu)}
\leq 
-\underbrace{\begin{bmatrix} x \\ u \end{bmatrix}^\top M_w
\begin{bmatrix} x \\ u \end{bmatrix}}_{w(Cx,-u)}. \label{eq:requirementWithAnsatzReal} 
\end{align}
The latter is strengthened by the LMI (\ref{eq:LMI}), which requires that $\mathrm{sym}( M_d(P))+M_w$ is negative semidefinite, i.e., (\ref{eq:requirementWithAnsatzReal}) shall hold for any $u\in \mathbb R^m$ instead of merely $u=-a(t,Cx)$. 
Even more, the underlying negative semidefiniteness in $\mathbb R$ of the real symmetric matrix 
\begin{align*}
\mathrm{sym}( M_d(P))+M_w \preceq 0
\end{align*}
 is also equivalent to its negative semidefiniteness in $\mathbb C$, in the sense of
\begin{align}
\forall z\in \mathbb C^{n+m}: \quad z^H\, \big(\mathrm{sym}( M_d(P))+M_w\big)\, z \leq 0. 
\end{align}
Consequently, the LMI (\ref{eq:LMI}) even means 
\begin{align}
&\forall x\in \mathbb C^n, \forall u\in \mathbb C^m: \nonumber \\
& 
\underbrace{\begin{bmatrix} x \\ u \end{bmatrix}^H \mathrm{sym}( M_d(P)) 
\begin{bmatrix} x \\ u \end{bmatrix}}_{\mathrm{Re}\big(2x^H P (Ax+Bu)\big)}
\leq 
-\underbrace{\begin{bmatrix} x \\ u \end{bmatrix}^H M_w
\begin{bmatrix} x \\ u \end{bmatrix}}_{=:w_{\mathbb C}(Cx,-u)}. \label{eq:LMImeaning} 
\end{align}
Concerning the underbrace on the left-hand side, note that the real part operator $\mathrm{Re}(\cdot)$ corresponds to the symmetrization\footnote{Let $N\in \mathbb R^{\nu\times \nu}$ and $z\in \mathbb C^\nu$. Then $z^H\mathrm{skw}(N)z\in \mathrm i \mathbb R$ and $z^H\mathrm{sym}(N) z \in \mathbb R$. Consequently, $\mathrm{Re}(z^H N z)=z^H\mathrm{sym}(N) z$.} of $M_d(P)$ from\footnote{Without the symmetrization in (\ref{eq:MdP}), a generalized definiteness term would have to be introduced. However, such generalizations for non-Hermitian matrices directly refer to the sign of the real part of the complex quadratic form, and thus yield the same result.} (\ref{eq:MdP}). 
The underbrace on the right-hand side introduces $w_{\mathbb C}:\mathbb C^p\times \mathbb C^m\to \mathbb R$, 
\begin{align}\label{eq:wcompl}
w_\mathbb C(\zeta,\alpha) &= \zeta^H \Pi_{\zeta\zeta}\zeta + \zeta^H\Pi_{\zeta\alpha} \alpha+\alpha^H\Pi_{\zeta\alpha}^\top \zeta+ \alpha^H \Pi_{\alpha\alpha}\alpha \nonumber\\
&=\mathrm{Re}(\zeta^H \Pi_{\zeta\zeta}\zeta + 2\zeta^H\Pi_{\zeta\alpha} \alpha+ \alpha^H \Pi_{\alpha\alpha}\alpha),
\end{align}
which is the complexification of $w:\mathbb R^p\times \mathbb R^m\to\mathbb R$ used in (\ref{eq:requirementWithAnsatzReal}), with $(\cdot)^\top$ replaced by $(\cdot)^H$, and again the real part unless being based on the symmetric form. 

The KYP lemma claims that, instead of considering all $x\in \mathbb C^n$ in (\ref{eq:LMImeaning}), it suffices to test whether (\ref{eq:LMImeaning}) holds for the special choice 
\begin{align} \label{eq:xf}
x_f:=(\mathrm i \omega I_n - A)^{-1} B u, \quad \omega\in \mathbb R 
\end{align}
 (assuming $\mathrm{det}(\mathrm i \omega I_n - A)\neq 0$ for all $\omega$, which, however, will again be weakened below). 
Notably, the left-hand side of (\ref{eq:LMImeaning}) with $P=P^H$ becomes zero if $x=x_f$ since 
\begin{align} \label{eq:xfBu}
Bu \stackrel{(\ref{eq:xf})}= (\mathrm i \omega I_n - A) x_f 
\end{align}
yields, in the left underbrace of (\ref{eq:LMImeaning}),
\begin{align}
\mathrm{Re}(x_f ^H P (A x_f + Bu))
\stackrel{(\ref{eq:xfBu})} = 
-\mathrm{Im} (x_f^H P\, x_f \,\omega) \stackrel{P=P^H}=0. \nonumber
\end{align} 
Altogether, (\ref{eq:LMImeaning}) restricted to $x=x_f$ from (\ref{eq:xf}) is 
\begin{align}
&\forall \omega\in \mathbb R, \forall u\in \mathbb C^m, x=(\mathrm i \omega I_n - A)^{-1} B u : 
\nonumber \\*
& \hspace{4.25cm}
0 \leq - w_{\mathbb C}(Cx,-u) \label{eq:KYPmotivation}
\end{align}
 and does not rely on the unknown $P$ anymore. The KYP lemma states (under weak assumptions) the following: {Inequality} 
(\ref{eq:KYPmotivation}) already ensures the existence of a real symmetric matrix $P$ that solves (\ref{eq:LMImeaning}) and thus the LMI (\ref{eq:LMI}). Due to Sec.~\ref{sec:solvabilityEquivalence}, {solvability} of the ARE (\ref{eq:ARE}) and the matrix equation (\ref{eq:matrixEquation}) are confirmed as well. 

The first argument of $w_\mathbb C$ in (\ref{eq:KYPmotivation}) can be written as
\begin{align}
Cx_f\stackrel{(\ref{eq:xf})}=& G(\mathrm i \omega) u
\quad \text{{with} }\label{eq:G}
\quad G(s)=C(sI-A)^{-1} B. 
\end{align}
Therefore, (\ref{eq:KYPmotivation}) is referred to as \textit{frequency-domain inequality} (FDI). 
 
\section{Solvability criterion: KYP}
The Kalman--Yakubovich--Popov (KYP) lemma relies on the FDI (\ref{eq:KYPmotivation}), requiring 
 \begin{align}
0&\leq -w_{\mathbb C}(C x_f , -u) = -w_{\mathbb C}(G(\mathrm i \omega) u , -u)
\nonumber \\
&\stackrel{(\ref{eq:wcompl})}=
- u^H \underbrace{ \begin{bmatrix}
G(\mathrm i \omega) \\ -I_m
\end{bmatrix}^H 
\begin{bmatrix} \Pi_{\zeta\zeta} & \Pi_{\zeta \alpha} \\
\Pi_{\zeta\alpha}^\top & \Pi_{\alpha\alpha}
\end{bmatrix}
\begin{bmatrix}
G(\mathrm i \omega) \\ -I_m
\end{bmatrix}
 }_{=:W_G^-(\mathrm i \omega)}
 u \label{eq:WGiw}
\end{align}
 for all $u\in \mathbb C^m$. It thus is a semidefiniteness requirement 
\begin{align}\label{eq:WGsemidefinite}
0\preceq -W_G^-(\mathrm i \omega), 
\end{align}
which is to be met for all (or almost all, see below) $\omega\in \mathbb R$. Based on the latter, the KYP lemma yields a statement on the solvability of the LMI, ARE, etc., from Sec.~\ref{sec:ApproachesV}.

Classical versions of the KYP lemma (see \cite{Gusev.2006} and references therein) assume controllability of $(A,B)$ (early versions even observability of $(A,C)$). In the present setting, $B$ and $C$ describe some perturbation structure, and---in contrast to a minimal realization of a given transfer function---there is no reason why $(A,B)$ should be controllable. We therefore use a version of the KYP lemma that comes with weaker assumptions. 
 
The first needed assumption is that, for at least one frequency point $\omega_0\in \mathbb R$, 
the definiteness in (\ref{eq:WGsemidefinite}) is even strict. 
\begin{assumption} \label{asm:WG}
$\exists \omega_0\in \mathbb R:  \lambda_{\max}(W_G^-(\mathrm i \omega_0))<0$.
\fine\end{assumption}
This assumption is weaker than $\Pi_{\alpha\alpha}\prec 0$ from Assumption~\ref{asm:PIaaNeg}: Due to the strict properness of $G(s)$, it holds $\lim_{\omega\to\infty} W_G^-(\mathrm i \omega)= \Pi_{\alpha\alpha}$, and thus Assumption \ref{asm:WG} is always satisfied if $\Pi_{\alpha\alpha}\prec 0$. 

A second assumption concerns {the matrix $A$ in (\ref{eq:ABaC})}. 
\begin{assumption} \label{asm:AeigImagCtrl}
$(A,B)$ does not have uncontrollable modes on the imaginary axis. 
\fine\end{assumption}
This assumption is weaker than Assumption \ref{asm:stabilizing} that ensures stabilizability due to the existence of $K_{\mathrm{stab}}$. Thus, it is also not restrictive in the present setting. 
Note that Assumption \ref{asm:AeigImagCtrl} does not rule out the possibility of controllable eigenvalues occurring on the imaginary axis. These lead to poles in $(\mathrm i \omega I_n -A)^{-1}B$ from (\ref{eq:xf}), and thus $W_G^-(\mathrm i \omega)$ is not defined on 
\begin{align} \label{eq:setGammaInf}
\Gamma_\infty:=\{\omega\in \mathbb R: \mathrm{det}(\mathrm i \omega I_n- A) = 0\}, 
\end{align}
even if only removable singularities occur in $W_G^-(\mathrm i \omega)$. 
It is already known from classical versions of the KYP lemma that it suffices\footnote{In contrast, $\Gamma_\infty$ cannot be ignored if it contains also uncontrollable eigenvalues, which is admissible in a KYP lemma for the strict LMI $M(P)\prec 0$, \cite[Lemma~1]{Balakrishnan.2003}. 
Therefore, in its most general form, the strict KYP lemma cannot be expressed only in terms of $G(\mathrm i \omega)$, but (\ref{eq:xfBu}) {is} used in place of (\ref{eq:xf}), cf.~\cite[Rem.~8]{Balakrishnan.2003}.} to consider only $\omega\in \mathbb R\setminus \Gamma_\infty$. 
\begin{lemma}[KYP lemma with weak assumptions] 
\label{lem:KYPwithoutctrl}
Suppose $W_G^-(\mathrm i \omega)$ from (\ref{eq:WGiw}) meets Assumption \ref{asm:WG}, and $A$ meets Assumption \ref{asm:AeigImagCtrl}. 
The solvability of the LMI, ARE, etc. described in Lemma~\ref{lem:solvabilityEquivalence} holds if and only if 
\begin{align} \label{eq:KYP}
\forall \omega\in \mathbb R\setminus \Gamma_{\infty}: \quad \lambda_{\max}(W_G^-(\mathrm i \omega))\leq 0.
\\*[-2em]\nonumber
\end{align}
\fine
\end{lemma}
\begin{proof}
See \cite[Thm.~1 and Rem.~1]{Churilov.1984} with a slight modification\footnote{It is already suggested in \cite[Cond.~($1_b$)]{Gusev.2006} that such a generalization is possible (without proof).}: In \cite{Churilov.1984}, it is assumed that $A$ does not have purely imaginary eigenvalues at all. However, the absence of purely imaginary eigenvalues of $A$ is only used in the proof when 
the existence of a unique solution to a Sylvester equation and the existence of a solution of a Lyapunov inequality are discussed, see \cite[p.~864]{Churilov.1984}. 
In both cases, it is argued that the involved coefficient matrix must not have purely imaginary eigenvalues. This matrix is given by the uncontrollable part of the controllability staircase form of $A-B(-\Pi_{\alpha\alpha})^{-1}\Pi_{\zeta\alpha}^\top C$. 
Since uncontrollable eigenvalues of $A$ are invariant under any feedback transform $A-B K_1$ with $K_1\in \mathbb R^{m\times n}$, the weakened Asm.~\ref{asm:AeigImagCtrl} suffices to render all arguments of the proof in \cite{Churilov.1984} still applicable. 
\end{proof}

\begin{remark}[Different usage of the KYP lemma]
As a result, in Lyapunov-based approaches, the KYP lemma is used to get a frequency-domain inequality (\ref{eq:KYP}) that ensures the feasibility of a given LMI, ARE, etc. In contrast, in Plancherel-Parseval-based approaches (IQCs, input-output-operator setting), cf.\ \cite{Megretski.1997}, the KYP lemma is used in the converse direction: to get an LMI that ensures that a given frequency-domain inequality holds (which in view of (\ref{eq:LMImeaning}) $\Rightarrow$ (\ref{eq:KYPmotivation}) is the simpler direction).
\end{remark}

\begin{remark}[Positive vs.\ negative feedback definition]
Instead of $G(s)=G_{u \,\mapsto \zeta}(s)$ from (\ref{eq:G}), which in terms of $\alpha=a(t,Cx)$ from (\ref{eq:ABaC}) refers to the negative-feedback input~${u=-\alpha}$, 
the transfer function 
\begin{align}
 G_{\alpha\,\mapsto \zeta}(s):=-G_{u \,\mapsto \zeta}(s) =- G(s)
\end{align}
for the positive-feedback input $\alpha$ can alternatively be considered. Then $W_G^-(\mathrm i \omega)$ from (\ref{eq:WGiw}) becomes {(cf.\ \cite{Megretski.1997})}
\begin{align*}
\\[-1em]
\smash{\begin{bmatrix} G(\mathrm i \omega) \\ -I_m \end{bmatrix}^H 
\Pi 
\begin{bmatrix} G(\mathrm i \omega) \\ -I_m \end{bmatrix}
= 
\begin{bmatrix} G_{\alpha\,\mapsto \zeta}(\mathrm i \omega) \\ I_m \end{bmatrix}^H 
\Pi 
\begin{bmatrix} G_{\alpha\,\mapsto \zeta}(\mathrm i \omega) \\ I_m \end{bmatrix}} .
\end{align*} 
\fine\end{remark}

\section{Frequency-domain results}\label{sec:FrequencyDomainResults}
It remains to apply the frequency-domain criterion (\ref{eq:KYP}) to the concrete sector constraints from Sec.~\ref{sec:sectorDescription}. The results are bounds on the respective sector-slope parameters $\gamma$, $\rho$ or $k_1$. If the sector bound is met, a quadratic function $V=x^\top P x$ from (\ref{eq:ansatzV}) exists that satisfies the defining {inequality} (\ref{eq:DVw}) from Sec.~\ref{sec:idea}. To be more precise, such a $V$ can then constructively be found by computing $P$ from the LMI, ARE, or matrix equation from Sec.~\ref{sec:ApproachesV}. These explicit bounds on the slope parameters will recover classical results such as the small-gain theorem, passivity conditions, and the circle criterion. 

\absatz{Sector: Linear norm bound}
With the matrix $\Pi$ from (\ref{eq:wLinearNormBound}) that describes the linear norm bound, $W_G^-(\mathrm i \omega)$ from (\ref{eq:WGiw}) becomes 
\begin{align}
W_G^-(\mathrm i \omega) &= \begin{bmatrix} G(\mathrm i \omega)\\-I \end{bmatrix}^H 
\begin{bmatrix} \gamma^2 I & 0 \\ 0 & -I \end{bmatrix}
 \begin{bmatrix} G(\mathrm i \omega)\\-I\end{bmatrix}
\nonumber\\*[0.5em]
&= \gamma^2(G(\mathrm i \omega))^H G(\mathrm i \omega) - I_m. 
\end{align}
Its maximum eigenvalue, which is needed in the KYP criterion (\ref{eq:KYP}),
\begin{align}
\lambda_{\max} (W_G^-(\mathrm i \omega))
&= \gamma^2 \underbrace{\lambda_{\max}( (G(\mathrm i \omega))^H G(\mathrm i \omega))}_{ \|G(\mathrm i \omega)\|_2^2}-1
\end{align}
relies on the spectral norm $\|G(\mathrm i \omega)\|_2$. Therefore, the supremum over all frequencies 
\begin{align}\label{eq:smallGainSupLambda}
\sup_{\omega\in \mathbb R} \lambda_{\max} (W_G^-(\mathrm i \omega))
 =
\gamma^2 \underbrace{\sup_{\omega\in \mathbb R} \|G(\mathrm i \omega)\|_2^2 }_{\|G\|_{H_\infty}^2} -1 
\end{align}
relies on the $H_\infty$-norm\footnote{which exists since $G$ is strictly proper and $A$ is Hurwitz.} of $G$ (due to (\ref{eq:Kstab}), $A$ is assumed to be Hurwitz). 
To meet\footnote{Due to $\Gamma_\infty=\emptyset$ (as $A$ is Hurwitz) and $G$ being strictly proper, the inequality in terms of the supremum is equivalent to the pointwise inequality (\ref{eq:KYP}).} the KYP criterion (\ref{eq:KYP}), $\sup_{\omega\in \mathbb R} \lambda_{\max} (W_G^-(\mathrm i \omega))$ from (\ref{eq:smallGainSupLambda}) must be nonpositive, 
\begin{align} 
\gamma^2 \|G\|_{H_\infty}^2 -1 \leq 0. 
\end{align}
Consequently, the following bound on $\gamma$ is obtained, which is a well-known corollary of the KYP lemma.
\begin{theorem}[Maximum linear norm bound]
Assume $A$ is Hurwitz. If the sector slope $\gamma$ of the linear norm bound constraint (\ref{eq:wLinearNormBound}) satisfies
\begin{align} \label{eq:gammamax}
\gamma\leq \frac 1 {\|G\|_{H_\infty}},
\end{align}
where $G(s)=C(sI-A)^{-1} B$, then a real symmetric matrix $P$ exists that solves the corresponding LMI, ARE, etc.\ from Lemma~\ref{lem:solvabilityEquivalence}. 
\fine\end{theorem}
If $\gamma$ satisfies (\ref{eq:gammamax}), then the sector $[-\gamma,\gamma]$ is admissible in the sense that a Lyapunov-like function can be found that proves due to Thm.~\ref{thm:ConvergenceStatement} attractivity of the zero equilibrium if the nonlinearity resides in that sector. 
\begin{remark}[Scalar case] 
If $p=m=1$, $\|G\|_{H_\infty}=\max_\omega \vert G(\mathrm i \omega)\vert$. Thus, the sector $[-\gamma,\gamma]$ is admissible if the Nyquist plot does not exceed a zero-centered disc in the complex plane with radius $1/\gamma$, see Fig.~\ref{fig:GainDisk}. As a result, the more restricted the Nyquist plot, the more robust is the nominal system, and the larger becomes the admissible sector for perturbations. 
\fine\end{remark}

\begin{figure}
\centering

\subfloat[Sector constraint on the {perturbation} graph $(\zeta,\tilde a(\zeta))$ 
\label{fig:gainSector}
]{
\includegraphics[]{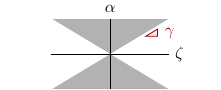}
}\hfill 
\subfloat[Corresponding constraint on the Nyquist plot of $G(\mathrm{i}\omega)$ 
\label{fig:GainDisk}
]{ 
\includegraphics[]{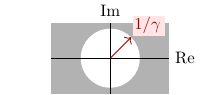}
}
\caption{Small-gain-theorem-like result 
}
\label{fig:NyquistGain}
\end{figure}

\begin{remark}[Relation to the small-gain theorem]\label{rem:smallGain}
If a quadratic offset function (\ref{eq:ellQuad}) is chosen in Thm.~\ref{thm:ConvergenceStatement}, the maximum gain of the perturbation must be strictly less than $\gamma$. Then, inequality (\ref{eq:gammamax}) means that the perturbation gain must be strictly less than the reciprocal of the $L_2$-{gain} of the system. The result thus mirrors the \textit{small-gain theorem}, cf.\ \cite[Sec.~III.2]{Desoer.2009}. It should be noted that this reciprocal relation can already be seen from comparing the LMI (\ref{eq:LMI}) for $\gamma$ with the LMI for the $L_2$-gain of the nominal system, before applying the KYP lemma, see \cite[Tab.~6.1]{Scholl.2024c}. For this LMI for the $L_2$-gain of $(A,B,C)$, the corollary of the KYP lemma, which correspondingly yields the $H_\infty$-norm of $G$ (or examines whether the gain equals $1$), is referred to as \textit{bounded-real lemma}~\cite{Boyd.1994}. \fine\end{remark} 
\begin{remark}[Relation to the complex stability radius]
The bound $1/\|G\|_{H_\infty}$ in (\ref{eq:gammamax}) coincides with the structured \textit{complex stability radius} of~$A$,~\cite{Hinrichsen.2005,Hinrichsen.1986b}. \fine\end{remark}

\absatz{Sector: Strict output passivity}
With the matrix $\Pi$ from (\ref{eq:wOutputPassivity}) that describes the strict output passivity sector, $W_G^-(\mathrm i \omega)$ from (\ref{eq:WGiw}) becomes 
\begin{align}
W_G^-(\mathrm i \omega) &= \begin{bmatrix} G(\mathrm i \omega)\\-I \end{bmatrix}^H 
\begin{bmatrix} 0 & \tfrac 1 2 I \\
\tfrac 1 2 I & -\rho I
\end{bmatrix} 
 \begin{bmatrix} G(\mathrm i \omega)\\-I\end{bmatrix}
\nonumber\\[0.5em]
&= -\tfrac 1 2 (G(\mathrm i \omega))^H -\tfrac 1 2 G(\mathrm i \omega) - \rho I . 
\end{align}
Its maximum eigenvalue
\begin{align}
\lambda_{\max} (W_G^-(\mathrm i \omega))
&= \underbrace{\lambda_{\max} \Big(\tfrac 1 2 \big((-G(\mathrm i \omega))^H -G(\mathrm i \omega)\big)\Big)}_{ \mu_2(-G(\mathrm i \omega))} - \rho \nonumber
\\[-2em]
\end{align}
relies on the logarithmic norm $\mu_2(-G(\mathrm i \omega))$. Therefore, the supremum over all frequencies (ignoring gaps from (\ref{eq:setGammaInf}))
\begin{align}\label{eq:PassivitySupLambda}
\sup_{\omega\in \mathbb R\setminus\Gamma_\infty} \lambda_{\max} (W_G^-(\mathrm i \omega))
 =
\underbrace{\sup_{\omega\in \mathbb R\setminus\Gamma_\infty} \mu_2(-G(\mathrm i \omega))}_{\stackrel{\mathrm{def}}=-\nu(G)} -\rho 
\end{align}
relies on the (thus formally defined) \textit{input passivity index}\footnote{cf.\ \texttt{nu=getPassiveIndex(sys,'input')} in Matlab} $\nu(G)$, \cite[eq.~(7)]{Zakeri.2022}. 
To meet 
the KYP criterion (\ref{eq:KYP}), $\sup_{\omega} \lambda_{\max} (W_G^-(\mathrm i \omega))$ from (\ref{eq:PassivitySupLambda}) 
must be nonpositive, i.e., 
\begin{align}\label{eq:passiveSupLambda}
-\nu(G) -\rho\leq 0, 
\end{align}
which yields the bound on $\rho$ in (\ref{eq:rhomin}) below. 
\begin{theorem}[Maximum upper bound]\label{thm:rhomin}
Let Assumption~\ref{asm:AeigImagCtrl} hold. If the upper sector slope $1/\rho$ of the strict output passivity constraint (\ref{eq:wOutputPassivity}) satisfies
\begin{align} \label{eq:rhomin}
\rho\geq -\nu(G),
\end{align}
where $-\nu(G)$ is defined in (\ref{eq:PassivitySupLambda}) and $G(s)=C(sI-A)^{-1} B$, then a real symmetric matrix $P$ exists that solves the corresponding LMI, ARE, etc.\ from Lemma~\ref{lem:solvabilityEquivalence}. 
\fine\end{theorem}

\begin{figure}
\centering

\subfloat[Sector constraint on the {perturbation} graph $(\zeta,\tilde a(\zeta))$ 
\label{fig:rhoSector}
]{
\includegraphics[]{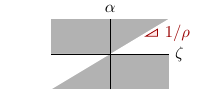}
}\hfill 
\subfloat[Corresponding constraint on the Nyquist plot of $G(\mathrm{i}\omega)$ 
\label{fig:rho}
]{ 
\includegraphics[]{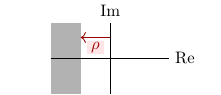}
}
\caption{Passivity-theorem-like result 
}
\label{fig:NyquistRho}
\end{figure}

Observe that $\lim_{\omega\to \infty} G(\mathrm i \omega)=0$ due to the strict properness of $G$ (no direct feedthrough in $(A,B,C,0)$). Thus, $\nu(G)\leq 0$, i.e., the right-hand side of (\ref{eq:rhomin}) is nonnegative.

\begin{remark}[Scalar case]
In the scalar case, $\mu_2(z)=\mathrm{Re}(z)$ for any $z\in \mathbb C$. Thus, if $p=m=1$, the input passivity index becomes $\nu(G)=\min_\omega \mathrm{Re}(G(\mathrm{i}\omega))$. A shortage of input passivity $\nu(G)<0$ of the nominal system $(A,B,C)$ means that the Nyquist plot extends into the left half plane. See Fig.~\ref{fig:rho}. Again, the more restricted the Nyquist plot, the more robust is the system and the larger becomes the admissible sector of perturbations. 
\end{remark}

\begin{remark}[Relation to passivity theorems]
Inequality (\ref{eq:rhomin}) means that the excess $\rho$ of output passivity of the nonlinearity must at least be the shortage $-\nu(G)$ of input passivity of the nominal system. The result thus mirrors a strict \textit{passivity theorem}. Similarly to Remark~\ref{rem:smallGain}, this relation can already be recognized from comparing the respective LMIs, see \cite[Tab.~6.1]{Scholl.2024c}. The corollary of the KYP lemma for the passivity of the nominal system (commonly for a vanishing $\nu(G)=0$) is referred to as \textit{positive-real lemma}~\cite{Boyd.1994}. 
\fine\end{remark}
\absatz{Sector: General sector constraints} 
For the general sector (\ref{eq:wCircle}), we focus on the case $p=m$ with $K_1=k_1 I_m$, and $K_2=k_2 I_m$, where {$k_1\leq k_2\in \mathbb R$}. We assume that $k_2$ is fixed and the smallest possible lower bound $k_1$ is of interest. 
 By (\ref{eq:wCircle}), the corresponding sector description relies on 
\begin{align*}
\Pi_{\zeta\zeta }=-k_1k_2 I_m , \quad \Pi_{\zeta\alpha}=\frac{k_1+k_2}{2}I_m, \quad \Pi_{\alpha\alpha }=-I_m. 
\end{align*} 

We seek a numerically tractable result for the lower sector bound $k_1$, similarly to Thm.~\ref{thm:rhomin}. To this end, we first transform the problem. Instead of the sector $[k_1,k_2]$ for $a(t,\zeta)$, we consider the anti-passive sector 
\begin{align}
[-(k_2-k_1),0]
\end{align}
 for the part $a^{\mathrm{II}}(t,\zeta)$ in
\begin{align*}
a(t,\zeta)=k_2 \zeta + a^{\mathrm{II}}(t,\zeta).
\end{align*}
 Correspondingly, the overall system (\ref{eq:ABaC}) becomes 
\begin{align}\label{eq:AII}
\dot x = \underbrace{(A-k_2BC)}_{=:A^{\mathrm{II}}} x - B a^{\mathrm{II}}(t,Cx), 
\end{align} 
and the sector $[-(k_2-k_1),0]$ for $a^{\mathrm{II}}(t,Cx)$ amounts to 
\begin{align*}
\Pi_{\zeta\zeta}^{\mathrm{II}}
=0_{m\times m},
\quad
 \Pi_{\zeta \alpha}^{\mathrm{II}}
=-\frac {k_2-k_1} 2 I_m , 
\quad
\Pi_{\alpha\alpha}^{\mathrm{II}}
=-I_m . 
\end{align*}
We confirm for the ARE formulation that the defining equation for the matrix $P$ for the transformed problem and for the original one are equivalent. 
\begin{lemma}[ARE equivalence]\label{lem:IIequivalence}
Replacing $A$ by $A^{\mathrm{II}}$ and $(\Pi_{\zeta\zeta},\Pi_{\zeta\alpha},\Pi_{\alpha\alpha})$ by $(\Pi_{\zeta\zeta}^{\mathrm{II}},\Pi_{\zeta\alpha}^{\mathrm{II}},\Pi_{\alpha\alpha}^{\mathrm{II}})$ does not alter the ARE (\ref{eq:ARE}). 
\fine\end{lemma} 
\begin{proof}
In the former case, the ARE (\ref{eq:ARE}) 
\begin{align*} 
 &PA+A^\top P = - C^\top \Pi_{\zeta\zeta} C - \tilde B(P)\, \tilde B^\top\!(P) 
\end{align*}
with $\tilde B(P)=(PB - C^\top \Pi_{\zeta\alpha})(-\Pi_{\alpha\alpha})^{-1/2}$ becomes 
\begin{align*} 
 PA+A^\top P&= k_1k_2 C^\top C \\*
&\;\;\; -\big(PB -\tfrac {k_1+k_2}{2} C^\top\big)\, \big(PB -\tfrac {k_1+k_2}{2} C^\top\big)^\top. 
\end{align*}
In the latter case, 
\begin{align*} 
 PA+A^\top P &- k_2 P BC - k_2 C^\top B^\top P \\*
& = -(PB + \tfrac{k_2-k_1}{2}C^\top)\,(PB + \tfrac{k_2-k_1}{2}C^\top)^\top. 
\end{align*}
Expanding both expressions shows that they coincide.
\end{proof}

\begin{theorem}[Minimum lower sector bound]
\label{thm:k1min}
Assume \smash{$(A^{\mathrm{II}},B)$} from (\ref{eq:AII}) has no uncontrollable modes on the imaginary {axis} (cf.~Assumption~\ref{asm:AeigImagCtrl}). Consider (\ref{eq:wCircle}) with $K_1=k_1 I_m$ and $K_2=k_2 I_m$, {$k_1\leq k_2\in \mathbb R$}. 
{Let $\nu(-G^{\mathrm{II}})=-\sup_\omega\mu_2(G^{\mathrm{II}}(\mathrm i \omega))$, with $G^{\mathrm{II}}(s)=C(sI-A^{\mathrm{II}})^{-1} B$ relying on $A^{\mathrm{II}}$ from (\ref{eq:AII}).}
{Assuming $\nu(-G^{\mathrm{II}})<0$,} if {the lower sector slope \smash{$k_1$} satisfies}
\begin{align} \label{eq:k1bound}
k_1\geq k_2+\frac 1 {\nu(-G^{\mathrm{II}})},
\end{align}
 then a real symmetric matrix $P$ exists that solves the corresponding LMI, ARE, etc.\ from Lemma~\ref{lem:solvabilityEquivalence}. 
\fine\end{theorem}
\begin{proof} 
With 
$(\Pi_{\zeta\zeta}^{\mathrm{II}},\Pi_{\zeta \alpha}^{\mathrm{II}},\Pi_{\alpha\alpha}^{\mathrm{II}})$ used in Lemma~\ref{lem:IIequivalence}, $W_{G^{\mathrm{II}}}^-(\mathrm i \omega)$ from (\ref{eq:WGiw}) becomes 
 \begin{align*}
&W_{G^{\mathrm{II}}}^-(\mathrm i \omega)=\tfrac 1 2 (k_2-k_1)\big ((G^{\mathrm{II}}(\mathrm i\omega))^H+ G^{\mathrm{II}}(\mathrm i\omega)\big)- I_m. 
\end{align*}
According to the KYP lemma, its maximum eigenvalue 
\begin{align*}
\lambda_{\max}(W_{G^{\mathrm{II}}}^-(\mathrm i \omega))
= 
 (k_2-k_1)\lambda_{\max}(\mathrm{He}(G^{\mathrm{II}}(\mathrm i \omega)))-1 \leq 0
\end{align*}
must be nonpositive for all $\omega$ (ignoring gaps (\ref{eq:setGammaInf})), i.e.,
\begin{align*}
k_2-k_1 \leq \frac 1 {\sup_{\omega\in \mathbb R\setminus \Gamma_\infty }\lambda_{\max}(\mathrm{He}(G^{\mathrm{II}}(\mathrm i \omega)))}. 
\\[-3.5em] \nonumber
\end{align*}
\end{proof}

\begin{figure}
\centering
\subfloat[Sector, $k_1<0<k_2$
\label{fig:CircleSector}
]{
\includegraphics[]{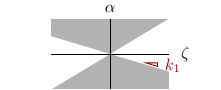}
}\hfill 
\subfloat[Nyquist, $k_1<0<k_2$
\label{fig:circleI}
]{
\includegraphics[]{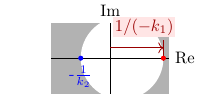}
}
\\[1em]
\hrule

\subfloat[Sector, $0<k_1<k_2$
\label{fig:CircleIISector}
]{
\includegraphics[]{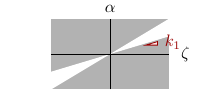}
}\hfill
\subfloat[Nyquist, $0<k_1<k_2$
\label{fig:circleII}
]{
\includegraphics[]{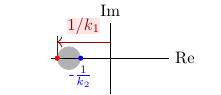}
}
\caption{Circle criterion 
}
\label{fig:CircleCriterion}
\end{figure}
\absatz{Graphical interpretation (Circle criterion)}
The criterion becomes, in the scalar case $p=m=1$, the standard circle criterion, see, e.g., \cite{Khalil.2002}. The latter can be recognized as follows. 
In terms of the real and imaginary part of $G(\mathrm i \omega)=\mathrm{Re}(G(\mathrm i \omega))+\mathrm i \, \mathrm{Im}(G(\mathrm i \omega))$, which are shown in the Nyquist plot, $W_G^-$ from (\ref{eq:WGiw}) is 
\begin{align*}
W_G^-(\mathrm i \omega)&=
\left[\begin{smallmatrix}
\mathrm{Re}(G(\mathrm i\omega)) \\ 
-1 \\
\mathrm{Im}(G(\mathrm i\omega)) 
\end{smallmatrix}\right]^{\!\top \!}\!
\left[\begin{smallmatrix} 
\Pi_{\zeta\zeta} 
& \Pi_{\zeta \alpha} 
&0\\
\Pi_{\zeta \alpha}^\top 
& \Pi_{\alpha\alpha} 
&0\\
0
&0
&\Pi_{\zeta\zeta} 
\end{smallmatrix} \right]
\left[\begin{smallmatrix}
\mathrm{Re}(G(\mathrm i\omega)) \\ 
-1 \\
\mathrm{Im}(G(\mathrm i\omega)) 
\end{smallmatrix}\right]. 
\end{align*} 
Let $p=m=1$, $k_1\neq 0$, and $k_2\neq 0$. Then the sector description (\ref{eq:condCircle}) can be rewritten as
\begin{align}
\tilde w(\zeta,\alpha)&:=\tfrac 1 {\pm k_1k_2} w(\zeta,\alpha)\nonumber \\
& =\pm (-\zeta+\tfrac 1 {k_1} \alpha)(\zeta-\tfrac 1 {k_2} \alpha) \geq 0, \label{eq:tildewCircleCrit}
\end{align}
with the sign $\pm$ such that 
\begin{align}\label{eq:CircleSign}
 \pm k_1k_2 > 0. 
\end{align}
 The coefficients in (\ref{eq:tildewCircleCrit}) (satisfying Assumption~\ref{asm:PIaaNeg}) are 
\begin{align}\label{eq:PiCircle}
\Pi_{\zeta\zeta}=\mp 1, \;\; \Pi_{\zeta\alpha}=\pm \frac{ \tfrac 1 {k_1}+ \tfrac{1}{k_2}} 2, \;\; \Pi_{\alpha\alpha}=
\mp \frac 1{k_1 k_2}. 
\end{align}
The resulting KYP inequality $W_G^-(\mathrm i \omega)\leq 0$ from (\ref{eq:KYP}) becomes in terms of $x=\mathrm{Re}(G(\mathrm{i}\omega))$ and $y=\mathrm{Im}(G(\mathrm{i}\omega))$ 
\begin{align} \label{eq:CircleCrit}
\left[\begin{smallmatrix}
x \\[0.25em]
-1 \\[0.25em]
y
\end{smallmatrix}\right]^{\!\top \!}\!
\left[\begin{smallmatrix} 
\Pi_{\zeta\zeta} 
& \Pi_{\zeta \alpha} 
&0\\
\Pi_{\zeta \alpha}^\top 
& \Pi_{\alpha\alpha} 
&0\\
0
&0
&\Pi_{\zeta\zeta} 
\end{smallmatrix} \right]
\left[\begin{smallmatrix}
x \\[0.25em] 
-1 \\[0.25em]
y
\end{smallmatrix}\right]
\leq 0. 
\end{align}
To see the circle criterion, compare the latter with the description of a disc in an $(x,y)\in \mathbb R^2$ plane. A closed disc with radius $r>0$ and shift $x_{\delta}^{}\in \mathbb R$ 
 can be described by $U(x,y)\leq 0$, 
\begin{align}
U(x,y)&= (x-x_\delta^{})^2+y^2-r^2 \nonumber
\\&= 
\left[\begin{smallmatrix} 
x\\-1\\y 
\end{smallmatrix}\right]^{\!\top} \!
\left[\begin{smallmatrix} 
1 & x_\delta^{} & 0 \\
x_\delta^{} & x_\delta^2-r^2 & 0 \\
0 & 0 & 1 
\end{smallmatrix}\right]
\left[\begin{smallmatrix} 
x\\-1\\y 
\end{smallmatrix}\right]. \label{eq:disc} 
\end{align} 

If $k_1k_2<0$, the lower right entry of (\ref{eq:CircleCrit}) is ${\Pi_{\zeta\zeta}=1}$ due to the sign convention from (\ref{eq:CircleSign}) in (\ref{eq:PiCircle}). Thus, if $k_1<0<k_2$, the KYP requirement (\ref{eq:CircleCrit}) means that the Nyquist plot must reside in the disc shown in Fig.~\ref{fig:circleI}. 

If $k_1k_2>0$, the lower right entry of (\ref{eq:CircleCrit}) is $\Pi_{\zeta\zeta}=-1$. Correspondingly, $-U(x,y)\leq 0$ from (\ref{eq:disc}) describes the closed complement of a disc.
Thus, if $0<k_1<k_2$, the KYP requirement (\ref{eq:CircleCrit}) means that the Nyquist plot must reside in the complement of a disc, as shown in Fig.~\ref{fig:circleII}. 

\absatz{Discussion} In contrast to the classical circle criterion, Thm.~\ref{thm:k1min} does not require a Nyquist plot. Even in the nonscalar case, only the maximum of the scalar real-valued function $\omega\mapsto \mu_2(G^{\mathrm{II}}(\mathrm i \omega))$ has to be determined (graphically/numerically, from implementations for the input passivity index $\nu(-G^{\mathrm{II}})$, or even analytically). If the LMI (\ref{eq:LMI}) is used, the bound on $k_1$ can be determined directly within the semidefinite programming {framework}.

\section{Example}\label{sec:Example}

To illustrate the procedure, we apply Thm.~\ref{thm:ConvergenceStatement} from Sec.~\ref{sec:attractivityConclusion} to a small example. If the sector constraint is not met globally, the theorem also provides a simple estimate of the domain of attraction. In the literature, related estimations are, e.g., proposed in \cite[p.~274]{Khalil.2002} for the circle criterion and \cite[p.~705]{Hinrichsen.2005} for the linear norm bound case. We construct a particularly well-behaved example. Nevertheless, also in more general cases, an improvement compared to the basic Lyapunov-equation-based estimate from \cite[p.~214]{Vidyasagar.2002}, \cite[p.~318]{Khalil.2002} (shown in Fig.~\ref{fig:Lyap}) can be expected, although the approach is similarly simple.

\begin{example}\label{ex:finalexample}
Fig.~\ref{fig:phaseportrait} shows the phase portrait of 
\begin{subequations}\label{eq:exampleg}
\begin{align}
\dot x_1&=x_2 \\*
\dot x_2&= x_1-2\mathrm{tanh}(x_1+x_2). 
\end{align}
\end{subequations}
The system has equilibria at $x_{e_1}=(0,0)$ and ${x_{e_{2,3}}=(\pm \beta,0)}$, where $\beta$ denotes the positive root of 
\begin{align}\label{eq:beta}
\beta-2\tanh(\beta)=0, \qquad \beta \approx 1.915. 
\end{align} 
A linearization of (\ref{eq:exampleg}) confirms that the zero equilibrium is locally asymptotically stable.

\begin{figure}
\footnotesize
\centering
 \subfloat[System (\ref{eq:exampleg}). Green trajectories~~~\newline belong to the domain of attraction. \label{fig:phaseportrait} 
]{
\includegraphics[]{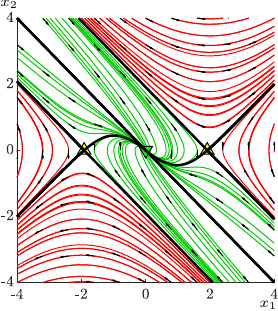}

}
 \subfloat[Basic approach for comparison \label{fig:Lyap}
]{
\includegraphics[]{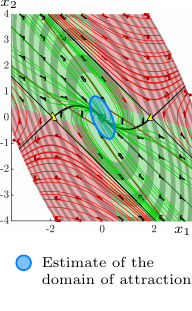}

}
\\
\subfloat[Sector \label{fig:exampleperturbation}
]{
\includegraphics[]{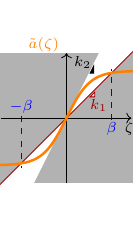}

}
 \subfloat[$V(x)$ \label{fig:ARE3d} 
]{
\includegraphics[]{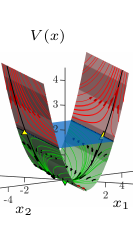}

}
 \subfloat[Result \label{fig:AREresult} 
]{
\includegraphics[]{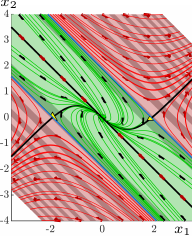}

}

\caption{ 
Example~\ref{ex:finalexample} is constructed such that {a} semidefinite Lyapunov-like function {can recover} the exact domain of attraction.
} \label{fig:example}
\end{figure}

As perturbation structure (\ref{eq:ABaC}) from Sec.~\ref{sec:problem} 
we choose
\begin{align}
\dot x= \underbrace{\begin{bmatrix} 0 & 1 \\ 1 & 0 \end{bmatrix}}_{A} x
-\underbrace{\begin{bmatrix} 0 \\ 1 \end{bmatrix}}_B
\tilde a\Big(\underbrace{\begin{bmatrix} 1 & 1 \end{bmatrix}}_{C} x\Big), \quad \tilde a (\zeta)= 2 \tanh(\zeta). 
\nonumber \\[-2em] \label{eq:exampleStructure}
\end{align} 
As sector constraint on $\tilde a$ from Sec.~\ref{sec:sectorDescription}, we choose $[k_1,k_2]$ with fixed upper sector bound $k_2>0$ and lower sector bound $k_1<k_2$ to be optimized, see Fig.~\ref{fig:exampleperturbation}. 
{Due to the concept of a $\mathcal K$-reduced sector from Sec.~\ref{sec:idea}}, we may choose $k_2=2$ as upper bound despite its coincidence with the tangent of $\tilde a$ at zero, cf.~(\ref{eq:ell}). Assumption~\ref{asm:stabilizing} is met as $K_{\mathrm{stab}}=2-\varepsilon$ with $0<\varepsilon\ll 1$ can be chosen.

From Thm.~\ref{thm:k1min}, we know that we will find a Lyapunov-like function $V(x)=x^\top P x$ whenever $k_1\geq k_1^{\mathrm{min}}$ from (\ref{eq:k1bound}). The latter relies on 
\begin{align*}
G^{\mathrm{II}}(s)=1/(s+1). 
\end{align*} 
Thus, $\nu(-G^{\mathrm{II}})=-\max_\omega\mathrm{Re}(G^{\mathrm{II}}(\mathrm i \omega))=-1$ yields 
\begin{align}
k_2=2 \quad \Rightarrow \quad k_1^{\mathrm{min}}\stackrel{(\ref{eq:k1bound}) }=k_2+1/\nu(-G^{\mathrm{II}})=1. 
\end{align} 
Let $k_1=k_1^{\mathrm{min}}$. 
{As shown in Fig.~\ref{fig:exampleperturbation}, the core perturbation function $\tilde a(\zeta)$ from (\ref{eq:exampleStructure}) is for any 
\begin{align*}
\zeta\in[-\beta,\beta]
\end{align*}
 in the sector $[k_1,k_2]$, relying on $\beta$ defined in (\ref{eq:beta})}. {Considering only $\zeta \in [-\tilde \beta,\tilde \beta]$ on a subinterval with $0<\tilde \beta<\beta$,} $\tilde a(\zeta)$ is even in the {$\mathcal K$-reduced} sector ({on a bounded domain, the graph of $\tilde a$ being in the interior of the sector already guarantees that some {class-$\mathcal K$} function $\ell(\zeta)=\kappa(\vert \zeta\vert)$ exists in (\ref{eq:ell}), cf.~Fig.~\ref{fig:boundary} and Rem.~\ref{rem:strictsector}}). Thus, {the set (\ref{eq:OmegaEll})} becomes 
\begin{align} \label{eq:Omegakappaexample}
\Omega_{\kappa(\vert\,\cdot\,\vert)}=\{ x\in \mathbb R^2: Cx \in[-\tilde \beta,\tilde \beta]\}, \quad 0<\tilde\beta<\beta.
\end{align}
{Due to the inequality chain (\ref{eq:DVnonposL}), this set already describes where the subsequently constructed Lyapunov-like function $V$ will have a partially negative definite derivative.}

According to Sec.~\ref{sec:ApproachesV}, we can use an LMI, ARE, or matrix equation to determine the (nonunique) matrix~$P$ from the quadratic ansatz (\ref{eq:ansatzV}). In this example, we choose the ARE (\ref{eq:ARE}),
\begin{align*}
&A^\top P + PA - (PB + S) R^{-1} (B^\top P + S^\top) + Q = 0 , \\
&
\begin{aligned}[t]
\Pi_{\zeta\zeta}&\stackrel{(\ref{eq:wCircle})}= -k_1k_2=-2 &\Rightarrow\; 
 Q&=C^\top \Pi_{\zeta\zeta}C=-2 \cdot 1_{2\times 2}, 
\\
\Pi_{\zeta\alpha}&\stackrel{(\ref{eq:wCircle})}=\tfrac{k_1+k_2} 2=\tfrac 3 2 &\Rightarrow\; 
S&=-C^\top \Pi_{\zeta\alpha}=-\tfrac 3 2 \cdot 1_{2\times 1}, 
\\ \Pi_{\alpha\alpha}&\stackrel{(\ref{eq:wCircle})}=-1 &\Rightarrow\; 
R&= \Pi_{\alpha\alpha}=-1. 
\end{aligned}
\end{align*}
The latter can be solved analytically from the Hamiltonian (cf. \cite[Rem.~6.2.14]{Scholl.2024c} with $k_1=k_1^{\min}+\varepsilon, \varepsilon \to 0^+$). We obtain $P=\tfrac 1 2 1_{2\times 2}$. {The resulting Lyapunov-like function
\begin{align*}
V(x)=x^\top P x= \tfrac 1 2 (x_1+x_2)^2 
\end{align*} 
is shown in Fig.~\ref{fig:ARE3d}.} 
{Notably, the set $\Omega_{\kappa(\vert\,\cdot\,\vert)}$ } from (\ref{eq:Omegakappaexample}), where any $\tilde \beta< \beta$ is admissible, 
already coincides with a sublevel set {$L_{\leq\vartheta}$} of $V$ with level $\vartheta=\tfrac 1 2 \tilde \beta^2$ since $V(x)=\tfrac 1 2 (Cx)^2$. {As a result of} Thm.~\ref{thm:ConvergenceStatement}~\ref{item:locallyStab}, any {strict sublevel set $L_{<\vartheta}$} with $\vartheta\leq \tfrac 1 2 \tilde \beta^2< \tfrac 1 2 \beta^2$ and thus {the union}
\begin{align*}
L_{\mathrm{res}}:=\bigcup_{\vartheta<\beta^2/2} L_{<\vartheta}=\{x\in \mathbb R^2: \vert Cx \vert < \beta\}
\end{align*} belongs to the domain of attraction, see {Fig.~\ref{fig:AREresult}}.

Since its boundaries coincide with the stable manifolds of the adjacent saddle equilibria, $L_{\mathrm{res}}$ is even the exact domain of attraction in this special example. 
\fine\end{example}

\section{Conclusion } \label{sec:Conclusion} 

The presented approach is based on a general defining {inequality} for the Lyapunov-like function $V$. Using a quadratic ansatz, a corresponding function $V$ can be computed via an LMI, an ARE, or a matrix equation. The approach relies on an inequality chain that ultimately leads to the sector-defining inequality. The introduction of the offset function $\ell$, which can be chosen as a general {class-$\mathcal K$} function{---thus giving rise to the concept of a $\mathcal K$-reduced sector---}is a key distinction from existing approaches to absolute stability. Moreover, a proof that relies on Barbalat's lemma combined with the BIBS and 0-CICS property of a stabilized system is presented. For the considered system class, it enables Lyapunov-like functions that may lack strict definiteness requirements. A passivity-index formulation is derived for the circle criterion. Furthermore, the article clarifies key interrelations between various concepts in the field of absolute stability.

\appendix
\section{Appendix} \label{sec:ProofConv}
The appendix contains {lemmas} used in Thm.~\ref{thm:ConvergenceStatement}.

 \begin{lemma}[Partial positive definiteness of $V$]\label{lem:VLowerBoundObsv}
Let {$V\colon \mathbb R^n\to \mathbb R$} be continuous and $V(0_n)=0$. 
If {its Dini derivative along solutions of a stabilized system $A-BK_{\mathrm{stab}}C$ (see Assumption~\ref{asm:stabilizing}) satisfies} 
\begin{align} \label{eq:DVKstabEllQuad}
{\exists k>0}, \forall x\in \mathbb R^n:\quad D_{A-BK_{\mathrm{stab}}C}^+V (x) \leq - k \|Cx\|_2^2, 
\end{align}
then $V$ admits the lower bound
\begin{align} \label{eq:partialPosDef}
\exists k_{1}>0,\forall x\in \mathbb R^n: \quad k_1\|x_\bot\|^2 \leq V(x) 
\end{align} 
 {in terms of} the $(A,C)$-observable part $x_\bot\in \mathcal U^\bot$ from~(\ref{eq:xBot}). 
\fine\end{lemma}
\begin{proof}
{Let $t\mapsto x(t)$ be a solution of 
$\dot x =(A-BK_{\mathrm{stab}}C)x$. Because of (\ref{eq:DVKstabEllQuad}), for every $t_1>0$, 
\begin{align}\label{eq:Vcomp}
V(x(t_1))- V(x(0)) 
\leq -\int_0^{t_1} k \|Cx(t)\|_2^2\,\mathrm d t. 
\end{align}}
Take the argument $x_0\in \mathbb R^n$ of $V(x_0)$ as {the} initial value {$x(0)=x_0$}. By the Hurwitz assumption, $\|x(t)\|\to 0$ as $t\to\infty$. Thus, due to the continuity of $V$, $V(x(t))\to V(0)=0$ as ${t\to\infty}$. Hence, 
\begin{align} 
V(x_0)&=-\Big( \lim_{t_1\to \infty} \underbrace{V(x(t_1))}_{\to 0} - V(\underbrace{x(0)}_{x_0})\Big) 
\nonumber\\[-0.5em]
&\stackrel{(\ref{eq:Vcomp})}\geq \int_0^\infty k \|Cx(t)\|_2^2\,\mathrm d t
\nonumber \\
&= k x_0^\top \Big(\underbrace{ \int_0^\infty \mathrm e^{A_{\mathrm{stab}}^\top t}C^\top C\mathrm e^{A_{\mathrm{stab}}t}\,\mathrm d t}_{W_{\mathrm o}} \Big)x_0, 
\label{eq:GramianLowerBound} \\[-2em] \nonumber
\end{align}
where $A_{\mathrm{stab}}:=A-BK_{\mathrm{stab}}C$ and $W_{\mathrm o }$ is the observability Gramian of $(A_{\mathrm{stab}},C)$. 
By the Courant--Fischer (min-max) theorem, the lower bound on $V$ from (\ref{eq:GramianLowerBound}) satisfies
\begin{align} 
x_0^\top \, W_{\mathrm o } \,x_0 \geq \lambda_{\min,\neq 0}(W_{\mathrm o }) \, \|x_\bot\|_2^2, 
\end{align}
where $\lambda_{\min,\neq 0}(W_{\mathrm o })$ denotes the smallest nonzero eigenvalue of $W_{\mathrm o }$, and where $x_\bot$ is the orthogonal projection of $x$ onto the orthogonal complement of the nullspace of $W_{\mathrm o }$. Moreover, the unobservable subspaces of $(A,C)$ and $(A_{\mathrm{stab}},C)$ with $A_{\mathrm{stab}}=A-BK_{\mathrm{stab}}C$ coincide. 
\end{proof}

\begin{lemma}[Positive invariance of a {strict} sublevel set]\label{lem:sublevelset}
{Let the Dini derivative of a continuous function $V\colon \mathbb R^n\to \mathbb R$ along solutions of (\ref{eq:ABaC}) be nonpositive on $\Omega_0\subset \mathbb R^n$, i.e.,} 
\begin{align}
&\forall x\in \Omega_0, \forall t\geq t_0: 
D_{(A\,\cdot\,-Ba(t,C\cdot))}^+ V (x) \leq 0.\label{eq:DVnonposLemma}
\end{align}
If {a level $\level> 0$ is chosen such that the sublevel set 
$L_{\leq\level}$ defined in (\ref{eq:sublevelset}) satisfies $L_{\leq \level}\subseteq \Omega_0$,}
 then the {strict} sublevel set $L_{< \level}$ defined in (\ref{eq:strictSublevelset}) 
is positively invariant, i.e., solutions {$x\colon[t_0,t_{\max})\to \mathbb R^n$} of (\ref{eq:ABaC}) satisfy
{\begin{align*}
x(t_0)\in L_{< \level} \quad \Longrightarrow \quad \forall t\in [t_0,t_{\max}): x(t)\in L_{<\level}. \\[-3.75em] \nonumber
\end{align*}}
\fine\end{lemma}
\begin{proof} 
{Let $x(t_0)\in L_{< \level}$, i.e., $V(x(t_0))< \level$. Suppose the solution leaves the strict sublevel set. Then, by continuity, there is a first time $t_1>t_0$ where $V(x(t_1))= \level$, while $V(x(t))< \level$ for $t\in [t_0,t_1)$. However, $V(x(t_1))>V(x(t_0))$ contradicts that $t\mapsto V(x(t))$ must be nonincreasing on $ [t_0,t_1]$ since the nonpositive Dini derivative $(\ref{eq:DVnonposLemma})$ applies to $x(t)\in L_{\leq \level}\subseteq \Omega_0$ for any $t\in [t_0,t_1]$.}
\end{proof}

\begin{lemma}[Bound on $x_\bot$]\label{lem:xBotBounded}
Let $V$ be a continuous function, and $\exists \kappa_1\in \mathcal K: \kappa_1(\|x_\bot\|)\leq V(x)$ with $\kappa_1(s)\to \infty$ as $s\to \infty$. 
If 
{$L_{< \level}$ defined in (\ref{eq:strictSublevelset})} is positively invariant, then for any initial value 
{$x(t_0)\in L_{< \level}$} 
the projection $x_\bot(t)$ is uniformly bounded for all $t\geq t_0$. 
\fine\end{lemma} 
\begin{proof}
Positive invariance guarantees that {$x(t)\in L_{< \level}$,} 
and thus 
{$V(x(t))< \level$,} 
for all $t\geq t_0$. Hence, 
{$\kappa_1(\|x_\bot(t)\|)\leq V(x(t))< \level$ implies $\|x_\bot(t)\| < \kappa_1^{-1}(\level)$.} 
\end{proof}

\begin{lemma}[Bound on $x$]\label{lem:boundedness}
Assume there exists a {$K_{\mathrm{stab}}\in \mathbb R^{m\times p}$} such that $A-BK_{\mathrm{stab}}C$ is Hurwitz {(Assumption~\ref{asm:stabilizing})}. Assume that $a(t,\zeta)$ is bounded for bounded $\zeta$, uniformly in $t$ {(Assumption~\ref{asm:aLocLip})}. Consider solutions {$x(\cdot)$} {of (\ref{eq:ABaC}), i.e.,} $\dot x=Ax-Ba(t,Cx)$. 
If the observable part $x_\bot(\cdot)$ (see (\ref{eq:xBot})) is bounded, then $x(\cdot)$ is bounded. 
\fine\end{lemma}
\begin{proof}
{Note that} only $C x=C(x_\bot + x_\|)=C x_\bot$ occurs in 
\begin{align}
\dot x&=Ax-Ba(t,Cx) \nonumber 
\\
&=(A-BK_{\mathrm{stab}}C) x +B (K_{\mathrm{stab}}Cx-a(t,Cx)), \label{eq:AsuCx}
\\
&= \underbrace{ (A-BK_{\mathrm{stab}}C)}_{A_\mathrm{stab}} x +B \underbrace{(K_{\mathrm{stab}}Cx_\bot-a(t,Cx_\bot))}_{=:u}. \nonumber 
\end{align}
Due to the Hurwitz assumption, $\dot x=A_\mathrm{stab} x+Bu$ shows the bounded-input bounded-state {(BIBS)} property \cite[Sec.~5.8.1]{Logemann.2014}. The input $u(\cdot)$ is bounded since $x_\bot(\cdot)$, and thus also $Cx_\bot(\cdot)$ is bounded, and $a(t,\zeta)$ is assumed to be bounded (uniformly w.r.t.\ $t$) for any bounded $\zeta$. 
\end{proof}

\begin{lemma}[Uniformly continuous $b_\kappa$]\label{lem:uniformlyContComposition}
Let $\kappa$ be a continuous function 
{and} $x\colon [t_0,t_{\max})\to \mathbb R^n; t\mapsto x(t)=x(t;t_0,x_0)$ be a solution of {(\ref{eq:ABaC}), i.e.,} $\dot x=Ax-Ba(t,Cx)$, where $a(t,\zeta)$ is locally Lipschitz in $\zeta$ (uniformly w.r.t.\ $t$). Assume $\|x(t)\|\leq c_x$ is upper bounded by some $c_x>0$ 
for all {$t\in [t_0,t_{\max})$}. 
Then $t\mapsto b_\kappa(t):= \kappa(\|C x(t)\|)$
 is uniformly continuous on $t\in [t_0,t_{\max})$. 
\fine\end{lemma}
\begin{proof}
The proof is along the lines of {\cite[Thm.~8.4]{Khalil.2002}}. As
\begin{align}
k\colon \mathbb R^n\supset D \to \mathbb R; \quad x \mapsto k(x)=\kappa(\|Cx\|)
\end{align} is a continuous function on a compact set $D:=\{x\in \mathbb R^n:\|x\|\leq c_x\}$, it is uniformly continuous. Moreover, $t\mapsto x(t)$ has a uniformly bounded derivative 
\begin{align*}
\|\dot x(t)\| &\leq \|A\| \|x(t)\|+\|B\| \|a(t,Cx(t))\|
\\ &= (\|A\| + L \|B\| \|C\|) \, c_x,
\end{align*}
{with} $L$ {being} the (time-invariant) Lipschitz constant from Assumption~\ref{asm:aLocLip} on the compact set {$\{\zeta\in \mathbb R^p: \zeta=Cx, x\in D\}$}, 
where $\|a(t,\zeta)\|\leq L\|\zeta\|$. {Thus, $x(\cdot)$} is also uniformly continuous. Finally, the composition $k\circ x $ of two uniformly continuous functions is uniformly continuous. 
\end{proof}

\begin{lemma}[Bound on $B_\kappa$]\label{lem:Bbounded}
Consider {(\ref{eq:ABaC}), i.e.,} $\dot x=Ax-Ba(t,Cx)$. Let $S$ be a positively invariant set. Assume $V$ is a continuous function bounded from below {by $V(\cdot)\geq V_{\min}\in \mathbb R$}, and, for any $x\in S$, 
\begin{align}\label{eq:DaVleq}
D_{A\cdot-Ba(t,C\cdot)}^+ V(x)\leq - \kappa(\|C x\|)
\end{align}
 {(see (\ref{eq:DVnonposL}) and (\ref{eq:ellK}))}. Then, for any {solution $x(\cdot)$ of (\ref{eq:ABaC}) with initial value} $x(t_0)\in S$, $B_\kappa(t)=\int_{t_0}^t \kappa(\|C x(\tau )\|)\,\mathrm d \tau$ is bounded. 
\fine\end{lemma}
\begin{proof} 
{Let $t\mapsto x(t)$ be a solution of 
$\dot x =Ax-Ba(t,Cx)$.} Because of 
\begin{align*}
\underbrace{V(x(t))}_{\geq V_{\min}}-V(x(t_0)) 
& 
\stackrel{(\ref{eq:DaVleq})}\leq - \int_{t_0}^t\kappa(\|C x(\tau )\|)\,\mathrm d \tau, 
\end{align*} 
for all $t\geq t_0$, there is an upper bound
\begin{align*}
\forall t\geq t_0: \quad \int_{t_0}^t \kappa(\|C x(\tau )\|)\,\mathrm d \tau \leq V(x(t_0))-V_{\min}.
\\[-3em] \nonumber
\end{align*}
\end{proof}

\begin{lemma}[Barbalat's lemma, {\cite[Lemma~8.2]{Khalil.2002}}] \label{lem:Barbalat}
Let $b:[t_0,\infty) \to \mathbb R$ and $B(t):=\int_{t_0}^t b(\tau)\,\mathrm d \tau$. If $b(t)$ is uniformly continuous on $t\in[t_0,\infty)$, then the existence of a finite limit $B_\infty \in \mathbb R$ implies 
\begin{align*}
\lim_{t\to\infty} B(t)=B_\infty \quad \Longrightarrow \quad 
\lim_{t\to\infty} b(t)=0. 
\\[-4em] \nonumber
\end{align*}
\fine\end{lemma}

\begin{lemma}[Convergence of $x$]\label{lem:convOfx}
Assume there exists a $K_{\mathrm{stab}}\in \mathbb R^{m\times p}$ such that $A_{\mathrm{stab}}=A-BK_{\mathrm{stab}}C$ is Hurwitz {(Assumption~\ref{asm:stabilizing})}. Consider solutions {$x(\cdot)$} {of (\ref{eq:ABaC}), i.e.,} $\dot x=Ax-Ba(t,Cx)$. 
If $ \|\zeta(t)\|\to 0$ implies $ \|a(t,\zeta(t))\|\to 0$ as $t\to\infty$ {(Assumption~\ref{asm:aLocLip})}, then 
\begin{align*}
 \lim_{t\to\infty}\|Cx (t)\|=0 \quad \Longrightarrow \quad \lim_{t\to\infty}\|x (t)\|=0 {.}
\\[-3.75em] \nonumber
\end{align*} 
\fine\end{lemma}
\begin{proof}
Since $\|a(t,\zeta(t))\|\to 0$, we have $\|u(t)\|\to 0$ in (\ref{eq:AsuCx}). 
Due to the Hurwitz assumption, $\dot x=A_\mathrm{stab} x+Bu$ shows the converging-input-converging-state property (0-CICS) \cite[Sec.~5.8.1]{Logemann.2014}. Thus, $\|u(t)\|\to 0$ implies $\|x(t)\|\to 0$. 
\end{proof}

{\small \sffamily
\bibliographystyle{IEEEtran}
\bibliography{Literatur}
}

\end{document}